\documentclass[journal]{IEEEtran}
\ifCLASSINFOpdf
\else
\fi
\usepackage{graphicx}
\usepackage{epstopdf}
\usepackage{amssymb}
\usepackage{color}
\usepackage{floatrow}
\usepackage{float}
\newfloatcommand{capbtabbox}{table}[][\FBwidth]
\usepackage{savesym}
\savesymbol{AND}
\savesymbol{c@algorithm}
\usepackage{algorithmic}
\usepackage{algorithm}

\newcommand{\Rmnum}[1]{\expandafter\@slowromancap\romannumeral #1@}
\makeatother
\usepackage{mathrsfs,amssymb}
\usepackage{mathrsfs}
\usepackage{bm}
\usepackage{amssymb}
\usepackage{mathrsfs}
\usepackage{subfigure}
\usepackage{mathrsfs}
\usepackage{subfigure}
\usepackage{multirow}
\usepackage{longtable,booktabs}
\usepackage[justification=centering]{caption}
\usepackage{cite}
\usepackage{amsmath}
\usepackage{threeparttable}
\usepackage{booktabs}
\usepackage{amssymb}
\usepackage{url}
\newtheorem{theorem}{Theorem}
\newtheorem{proposition}{Proposition}

\newtheorem{definition}{Definition}
\newtheorem{lemma}{Lemma}
\newtheorem{example}{Example}
\newtheorem{remark}{Remark}

\newenvironment{proof}{\textit{Proof}:}{\hfill $\blacksquare$\par}
\definecolor{darkgreen}{cmyk}{0.91,0,0.88,0.32}
\definecolor{orange}{cmyk}{0.0,.5,1,0.1}
\definecolor{purple}{cmyk}{0.2,1,0,0.1}
\definecolor{pink}{cmyk}{0,0.4,0.4,0}
\definecolor{gray}{cmyk}{0.2,0.2,0.2,0}
\definecolor{greener}{cmyk}{0.91,0,0.88,0.2}

\usepackage{graphicx}          % Include this line if your
\usepackage[english]{babel}	
\usepackage{fancyhdr}

\begin{document}

\IEEEspecialpapernotice{{\rm This work has been submitted to the IEEE for possible publication. Copyright may be transferred without notice, after which this version may no longer be accessible.}}
%\pagestyle{fancy}
%\fancyhead[C]{\large This work has been submitted to the IEEE for possible publication. Copyright may be transferred without notice, after which this version may no longer be accessible.}
%\title{Verification of Nonblockingness in Bounded Petri Nets: A Novel Semi-Structural Approach}
\title{Verification of Nonblockingness in Bounded Petri Nets With Minimax Basis Reachability Graphs}


%\author{Chao Gu, Ziyue Ma, Zhiwu Li and Alessandro Giua% <-this % stops a space
\author{Chao~Gu,~\IEEEmembership{Student Member,~IEEE}
        ~Ziyue Ma,~\IEEEmembership{Member,~IEEE}
        ~Zhiwu Li,~\IEEEmembership{Fellow,~IEEE}
        and~Alessandro\\ Giua,~\IEEEmembership{Fellow,~IEEE}
\thanks{This work is partially supported by the National Key R\&D Program of China under Grant 2018YFB1700104, the National Natural Science Foundation of China under Grants 61873342 and 61472295, the Shaanxi Provincial Natural Science Foundation under Grant No. 2019JQ-022, the Fundamental Research Funds for the Central Universities under Grant JB190407, the Science and Technology Development Fund, MSAR, under Grant No. 0012/2019/A1, and the Fund of China Scholarship Council under Grant No. 201806960056 (\textit{Corresponding author: Zhiwu Li}).}% <-this % stops a space
\thanks{C. Gu is with the School of Electro-Mechanical Engineering, Xidian University, Xi'an 710071, China, and also with DIEE, University of Cagliari, Cagliari 09124, Italy
        {\tt\small cgu1992@stu.xidian.edu.cn}}%
\thanks{Z. Ma is with the School of Electro-Mechanical Engineering, Xidian University, Xi'an 710071, China
        {\tt\small maziyue@xidian.edu.cn}}
\thanks{Z. Li is with the School of Electro-Mechanical Engineering, Xidian University, Xi'an 710071, China, and also with the Institute of Systems Engineering, Macau University of Science and Technology, Macau {\tt\small zhwli@xidian.edu.cn}}
\thanks{A. Giua is with DIEE, University of Cagliari, Cagliari 09124, Italy {\tt\small giua@unica.it}}
}

\maketitle

%\thispagestyle{fancy} ? ? ? ?
%\fancyhead{} ? ? ? ? ? ? ? ? ? ?
%\chead{This work has been submitted to the IEEE for possible publication. Copyright may be transferred without notice, after which this version may no longer be accessible.}
%\cfoot{\quad}
%\renewcommand{\headrulewidth}{0pt} ? ?
%
%\renewcommand{\footrulewidth}{0pt}
\thispagestyle{empty}
\pagestyle{empty}


\begin{abstract}                          % Abstract of not more than 200 words.
This paper proposes a semi-structural approach to verify the \textit{nonblockingness} of a Petri net.
We construct a structure, called \textit{minimax basis reachability graph} (minimax-BRG): it provides an abstract description of the reachability set of a net while preserving all information needed to test if the net is \textit{blocking}.
We prove that a bounded deadlock-free Petri net is \textit{nonblocking} if and only if its minimax-BRG is \textit{unobstructed}, which can be verified by solving a set of integer constraints and then examining the minimax-BRG.
For Petri nets that are not deadlock-free, one needs to determine the set of deadlock markings. This can be done with an approach based on the computation of \textit{maximal implicit firing sequences} enabled by the markings in the minimax-BRG.
%{\red The approach we developed achieves significant practical efficiency, as shown by means of numerical simulations.}
%{\blue The approach we developed does not require the construction of the reachability graph and achieves practical efficiency, as shown by means of numerical simulations.}
The approach we developed does not require the construction of the reachability graph and has wide applicability.
%The approach we developed does not require exhaustive exploration of the state space and therefore achieves significant practical efficiency, as shown by means of numerical simulations.
\end{abstract}

%\end{frontmatter}

\section{Introduction}\label{Section1}

As discrete event models, Petri nets are commonly used in the framework of \textit{supervisory control theory} (SCT) \cite{SL, c27, luo2009supervisor, wang2013design}.
%However, the application of SCT using automaton models is hindered by the \textit{state explosion} problem.
From the point of view of computational efficiency, Petri nets have several advantages over simpler models such as automata \cite{c1, c27, Murata}: since states in Petri nets are not explicitly represented in the model in many cases, and structural analysis and linear algebraic approaches can be used without exhaustively enumerating the state space of a system.
%{\red It should be noted, however, many analysis techniques for Petri nets are based on the construction of the \textit{reachability graph}, which suffers from the same state explosion problem typical of automata. To take advantage of this compact representation, other techniques should be exercised.}

A suite of supervisory control approaches in discrete event systems focuses on an essential property, namely \textit{nonblockingness} \cite{RW, c11}. As defined in \cite{RW}, nonblockingness is a property prescribing that all reachable states should be \textit{co-reachable} to a set of \textit{final states} representing the completions of pre-specified tasks.
%A system that is not nonblocking may not successfully finish its tasks, which incurs unnecessary costs on-stream.
Consequently, to verify and ensure the nonblockingness of a system is a problem of primary importance in many applications and should be addressed with state-of-the-art techniques.
%However, as is known, to ensure nonblockingness remains a challenge due to the optimal nonblocking supervisory control is proved to be NP-hard \cite{gohari2000complexity}.

%{\orange ***In the literatures of Petri nets, how people have done to solve/enforce nonblockingness, and why it is difficult to be solved.***}
The nonblockingness verification (NB-V) problem in automata can be solved in a relatively straightforward manner. The authors in \cite{clin} address several sufficient conditions for nonblockingness; however, they are not very suitable for systems that contain complex feedback paths. In \cite{leduc2000hierarchical, leduc2005hierarchical}, a method called \textit{hierarchical interface-based supervisory control}, which consists in breaking up a plant into two subsystems and restricting the interaction between them, is developed to verify if a system is nonblocking.
Based on the \textit{state tree structure}, \cite{c24} studies an efficient algorithm for nonblocking supervisory control design in reasonable time and memory cost.
To mitigate the \textit{state explosion} problem, in the framework of \textit{compositional verification} an abstraction approach is proposed in \cite{mohajerani2016framework} to verify discrete event systems modelled by extended finite-state machines (EFSMs) and such verification approach is typically designed for large models consisting of several EFSMs that interact both via shared events and variables.
Based on the automaton \textit{abstraction} technique, the work in \cite{su2010aggregative} presents an \textit{aggregative synthesis approach} to obtain nonblocking supervisors in a distributed way.

Using Petri net models, the works in \cite{c1, c11} study NB-V and enforcement from the aspect of \textit{Petri net languages}; however, these methods rely on the construction and analysis of the \textit{reachability graph}, which is practically inefficient.
A different approach based on the \textit{theory of regions}\cite{uzam2002optimal} was used in \cite{ghaffari2003design} to design a maximally permissive controller ensuring the nonblockingness of a system; however, it still requires an exhaustive enumeration of the state space.
For a class of Petri nets called \textit{G-systems}, \cite{zhao2013iterative} reports a deadlock prevention policy that can usually lead to a nonblocking supervisor with high computational efficiency but cannot guarantee maximally permissive behavior.
A nonblocking and maximally permissive supervisor with a distributed architecture is designed in \cite{hu2015maximally} and modelled by a class of Petri nets namely $BS^{4}R$.

As is known, the difficulty of enforcing nonblockingness lies in the fact that the optimal nonblocking supervisory control problem is \textit{NP-hard} \cite{gohari2000complexity}.
Moreover, the problem of efficiently verifying nonblockingness of a Petri net without constructing its reachability graph remains open to date.
By this motivation, in this paper, we aim to develop a method to cope with the NB-V problem in Petri nets.

A state-space abstraction technique in Petri nets, called  \emph{basis reachability graph (BRG)} approaches, was recently proposed in \cite{c8,Basis}. In these approaches, only a subset of the reachable markings, called \emph{basis markings}, are enumerated.
This method can be used to solve \emph{marking reachability} \cite{c3}, \emph{diagnosis} \cite{c8,Basis} and \emph{opacity} problems \cite{c10} efficiently. Thanks to the BRG, the state explosion problem can be mitigated and the related control problems can be solved efficiently. The BRG-based methods are \emph{semi-structural} since only basis markings are explicitly enumerated in the BRG while all other reachable markings are abstracted by linear algebraic equations.
%{\blue Inspired by the classical BRG-based methodology, in this paper, we will develop a novel semi-structural approach to verify the nonblockingness of a Petri net.}

%Supervisory control theory of discrete event systems \cite{RW,RW2,RW3,c24} involves a number of important notions, such as \textit{controllability}, \textit{observability}, and \textit{nonblockingness}. Controllability is a property that in a controlled system a supervisor never attempts to disable an uncontrollable event. The work in \cite{c2} {\red uses} BRG to verify {\red the} controllability of a Petri net, where the controllability of each basis marking is determined by solving an \emph{integer linear programming problem} (ILPP). Such a way of controllability verification is practically efficient since it does not require an exhaustive enumeration of the state space.
%
%In order to make appropriate control choices, it is often required that the controller should have the overall knowledge of system states. However, direct state knowledge is usually {\red not available} due to the limitations of sensing.
%Fundamentally identified in \cite{Observability}, in discrete event systems, the property of observability allows an observer to estimate states that cannot be measured \cite{Observability3, Observability4}.

On the other hand, in our previous work \cite{Gu} we show that the standard BRG cannot be directly used to solve the NB-V problem due to the possible presence of \textit{livelocks} and \textit{deadlocks}.
In particular, livelocks describe an undesirable non-dead repetitive behavior such that the system is bound to evolve along a particular subset of its reachability space.
%Once a marking in a livelocked component is reached, the future evolution will remain within this component.
Thus, a Petri net is blocking if a livelock that contains no final markings is reachable.
However, the set of markings that form a livelock is usually hard to characterize and is not encoded in the classical BRG of the system.
%In our preliminary work in \cite{Gu}, we proved that for deadlock-free Petri nets, the {\red NB-V problem} can be {\red addressed} by constructing a structure called \textit{expanded BRG} and checking nonblockingness of each node it contains. However, the efficiency of this approach needs to be further improved.
As a countermeasure, preliminary results are presented in \cite{Gu} to show how it is possible to tailor the BRG to detect livelocks.
In more detail, a structure named the \textit{expanded BRG} is proposed, which expands the BRG so that all markings in $R(N, M_0)$ reached by firing a sequence of transitions ending with an explicit transition are included. The set of markings in an expanded BRG is denoted as the \textit{expanded basis marking} set $\mathcal{M_{B_E}}$.
However, this approach presents two major drawbacks. First, it only applies to deadlock-free nets, which is an undesirable restriction considering that dead non-final markings are one of the causes of blockingness. Second, while the expanded BRG can abstract part of the reachability set, its size can still be very large and its practical efficiency needs to be improved.

When a system is not deadlock-free, a dead marking in the state space characterizes a condition from which the system cannot further advance \cite{c20}. If there exists a dead marking that is not final (we call such a state a \textit{non-final deadlock}), the system is blocking.

Inspired by the classical BRG-based methodology, in this paper, we develop a semi-structural approach to tackle the NB-V problem. The contribution consists of three aspects:

\begin{itemize}
   \item [$-$]We propose a structure called \textit{minimax basis reachability graph} (minimax-BRG). In minimax-BRGs, only part of the state space, namely \textit{minimax basis markings}, is encoded and all other markings can be characterized as the integer solutions of a linear constraint set.
   \item [$-$] Owing to properties of the minimax-BRG, when a bounded Petri net is known to be deadlock-free, we prove that it is nonblocking if and only if its minimax-BRG consists of all nonblocking nodes (such a minimax-BRG is said to be \textit{unobstructed}), which can be verified by solving a set of integer constraints and then examining the minimax-BRG.
  % \item [$-$] finally, we generalize the results to systems that are not deadlock-free.
%       Our approach does not require the construction of the reachability graph in prior and achieves practical efficiency in the considered cases, as shown by numerical simulations.
   \item [$-$] We generalize the results to arbitrary bounded Petri nets (not necessarily be deadlock-free) and propose a necessary and sufficient condition for NB-V.
       Numerical results demonstrate the proposed approach.
       %Our approach does not require constructing the reachability graph and has wide applicability since it does not depend on specific substructures of the Petri net.
%   Hence, the {\red NB-V problem} of nets that are not deadlock-free can be {\red addressed} by first determining the non-final deadlocks followed by checking its minimax-BRG. The approach we developed does not require exhaustive exploration of the state space and therefore achieves significant practical efficiency.
\end{itemize}

The rest of the paper is organized as follows. Some basic concepts and formalisms used in the paper are recalled in Section \uppercase\expandafter{\romannumeral2}. Section \uppercase\expandafter{\romannumeral3} dissects the NB-V problem.
Section \uppercase\expandafter{\romannumeral4} introduces the minimax-BRG.
Section \uppercase\expandafter{\romannumeral5} investigates how minimax-BRGs can be applied to solving the NB-V problem.
Numerical analyses are given in Section \uppercase\expandafter{\romannumeral6}, while discussions are reported in Section \uppercase\expandafter{\romannumeral7}.
Conclusions and future work are given in Section \uppercase\expandafter{\romannumeral8}.
%Section $5$ develops a novel structure named the minimax-BRG and exposes a sufficient and necessary condition for nonblockingness verification of a deadlock-free system.
%In Section \uppercase\expandafter{\romannumeral5}, we generalize the above results for the systems that are not deadlock-free.
%Numerical analyses are given in Section \uppercase\expandafter{\romannumeral6}.

%
%Section \uppercase\expandafter{\romannumeral7} draws conclusions and discusses future work.

\section{Preliminaries}\label{sec3/2}
%In this section, we recall the main notions related to automata, Petri nets \cite{Murata}, and basis markings \cite{c3,c8,Basis} used in the paper.
\subsection{Automata and Petri nets}
An automaton \cite{RW3} is a five-tuple $A=(X,\Sigma,\eta,x_0,X_m)$, where $X$ is a set of \textit{states}, $\Sigma$ is an alphabet of \textit{events}, $\eta : X \times \Sigma\rightarrow X$ is a \textit{state transition function}, $x_0\in X$ is an \textit{initial state} and $X_m\subseteq X$ is a set of \textit{final states} (also called \textit{marker states} in \cite{RW}).
$\eta$ can be extended to a function $\eta : X \times \Sigma^{*}\rightarrow X$.

A state $x\in X$ is \emph{reachable} if $x = \eta(x_0, s)$ for some $s\in \Sigma^*;$ it is \emph{co-reachable} if there exists $s^{\prime}\in \Sigma^*$ such that $\eta(x, s^{\prime})\in X_m$.
An automaton is said to be \emph{nonblocking} if every reachable state is co-reachable.

%
%\subsection{Petri Nets}\label{PNbasic}
A Petri net \cite{Murata} is a four-tuple $N=(P,T,Pre,Post)$, where $P$ is a set of $m$ \textit{places} (graphically represented by circles) and $T$ is a set of $n$ \textit{transitions} (graphically represented by bars).
%$P$ and $T$ are finite and disjoint sets, i.e., $P\neq\emptyset$, $T\neq\emptyset$, and $P\cap{T}=\emptyset$.
$Pre: P\times T\rightarrow \mathbb{N}$ and $Post: P\times T\rightarrow \mathbb{N}$ ($\mathbb{N}=\{0, 1, 2, \cdots\}$) are the \textit{pre}- and \textit{post}- \textit{incidence functions} that specify the \textit{arcs} directed from places to transitions, and vice versa in the net, respectively.
%and are represented as matrices in $\mathbb{N}^{m\times n}$ ($\mathbb{N}=\{0, 1, 2, \cdots\}$).
The \emph{incidence matrix} of $N$ is defined by $C=Post-Pre$.
A Petri net is \textit{acyclic} if there are no directed cycles in its underlying digraph.
%A Petri net $N=(P, T, Pre, Post)$ is said to be \textit{ordinary} if for all $p\in P$ and $t\in T$, $Pre(p, t)\in \{0, 1\}$ and $Post(p, t)\in \{0, 1\}$; otherwise it is said to be a \textit{generalized net}.

Given a Petri net $N=(P,T,Pre,Post)$ and a set of transitions $T_x\subseteq T$, the \textit{$T_x$-induced sub-net} of $N$ is a net resulting by removing all transitions in $T\setminus T_x$ and corresponding arcs from $N$, denoted as $N_x=(P,T_x,Pre_x,Post_x)$ where $T_x\subseteq T$ and $Pre_x$ ($Post_x$) is the restriction of $Pre$ ($Post$) to $P$ and $T_x$.
The incidence matrix of $N_x$ is denoted by $C_x = Post_x-Pre_x$.

A \emph{marking} $M$ of a Petri net $N$ is a mapping: $P\to\mathbb{N}$ that assigns to each place of a Petri net a non-negative integer number of \textit{tokens}. The number of tokens in a place $p$ at a marking $M$ is denoted by $M(p)$.
A Petri net $N$ with an initial marking $M_0$ is called a \textit{marked net}, denoted by $\langle N, M_0\rangle$.

%Let $p\in{P}$ be a place of a Petri net $N$. It is marked at $M$ if $M(p)>0$.
%A set of places $D\subseteq{P}$ is marked at $M$ if at least one place in $D$ is marked, viz., $\exists p\in{D}, M(p)>0$.
%
For a place $p\in P$, the \textit{set of its input transitions} is defined by $^{\bullet}p=\{t\in T\mid Post(p,t)>0\}$ and the \textit{set of its output transitions} is defined by $p^{\bullet}=\{t\in T\mid Pre(p,t)>0\}$. The notions for $^{\bullet}t$ and $t^{\bullet}$ are analogously defined.

%Given a place (transition) $p\ (t)$, the elements in its preset are called the Pre-transitions (Pre-places) of $p\ (t)$, while the Postset of $p\ (t)$ is named as Post-transitions (Post-places).
A transition $t\in T$ is \emph{enabled} at a marking $M$ if $M\geq Pre(\cdot, t)$\footnote{We use $A(\cdot, x)$ ($A(x, \cdot)$) to denote the column (row) vector corresponding to the element $x$ in matrix $A$.}, denoted by $M[t\rangle$. If $t$ is enabled at $M$, the \emph{firing} of $t$ yields marking $M^{\prime}=M+C(\cdot, t)$, which is denoted as $M[t\rangle M^{\prime}$. A marking $M$ is \textit{dead} if for all $t\in T$, $M\ngeqslant Pre(\cdot, t)$.
%\footnote{Here $Pre(\cdot, t)$ and $C(\cdot, t)$ denote, respectively, the column vector corresponding to transition $t$ in matrix $Pre$ and $C$.}

%A Petri net is said to be free of self-loop if there do not exist a place $p$ and a transition $t$ such that $(p,t)\in F$ and $(t,p)\in F$. A self-loop-free Petri net can be represented by an incidence matrix $[N](p,t)=W(t,p)-W(p,t)$ that is an integer matrix indexed by $P$ and $T$.
Marking $M^{\prime}$ is \emph{reachable} from $M_{1}$ if there exist a sequence of transitions $\sigma=t_{1}t_{2}\cdots t_{n}$ and markings $M_{2},\cdots, M_{n}$ such that $M_{1}[t_{1}\rangle M_{2}[t_{2}\rangle\cdots M_{n}[t_{n}\rangle M^{\prime}$ holds. When $\sigma = \epsilon$, where $\epsilon$ denotes the empty sequence, then it holds that $M[\sigma\rangle M$.
We denote by $T^*$ the set of all finite sequences of transitions over $T$.
Given a transition sequence $\sigma\in T^{*}$, $\varphi: T^{*}\rightarrow \mathbb{N}^{n}$ is a function that associates to $\sigma$ a vector $\textbf{y}=\varphi(\sigma)\in \mathbb{N}^{n}$, called the \textit{firing vector} of $\sigma$, i.e., $\textbf{y}(t) = k$ if transition $t\in T$ appears $k$ times in $\sigma$. In particular, it holds that $\varphi(\epsilon) = \textbf{0}$.
Let $\varphi^{-1}: \mathbb{N}^{n}\rightarrow T^{*}$ be the inverse function of $\varphi$, namely for $\textbf{y}\in \mathbb{N}^{n}$, $\varphi^{-1}(\textbf{y}):=\{\sigma\in T^{*}| \varphi(\sigma)=\textbf{y}\}$.

The set of markings reachable from $M_{0}$ is called the \emph{reachability set} of $\langle N, M_0\rangle$, denoted by $R(N, M_{0})$.
A marked net $\langle N, M_0\rangle$ is said to be \emph{bounded} if there exists an integer $k\in \mathbb{N}$ such that for all $M\in R(N, M_0)$ and for all $p\in P$, $M(p)\leq k$ holds.

\begin{proposition}{\rm\cite{c8,Murata}}\label{ProX}
{\rm Given a marked net $\langle N, M_0\rangle$ where $N$ is acyclic, $M\in R(N, M_0)$, $M^{\prime}\in R(N, M_0)$ and a firing vector $\textbf{y}\in \mathbb{N}^n$, the following holds:
\begin{center}
$\hspace{5mm}M^{\prime}=M+C\cdot \textbf{y}\geq \textbf{0}\Leftrightarrow (\exists \sigma\in \varphi^{-1}(\textbf{y}))\ M[\sigma\rangle M^{\prime}.\hfill\square$
\end{center}}
\end{proposition}
Proposition \ref{ProX} shows that in acyclic nets, reachability can be characterized (necessary and sufficient condition) in simpler algebraic terms.

%A non-empty place subset $S\subseteq{P}$ is a \textit{siphon} if $^\bullet{S}\subseteq{S}^\bullet$. A siphon $S$ is minimal if the removal of any place from $S$ makes the fallacy of $^\bullet{S}\subseteq{S}^\bullet$.
%A siphon can also be described by its \textit{characteristic vector} $\textbf{s}\in\{0, 1\}^{m}$ such that $s_i=1$ if $p_i\in S$, otherwise $s_i=0\ (i\in\{1, 2, \cdots, m\})$.

%%%%%%%%%%%%%%%%%%%%%%%%%%%%%%%%%%%%%%%%%%
%%%%%%%%%%%%%%%%%%%%%%%%%%%%%%%%%%%%%%%%%%
%%%%%Original Characterization of M_F%%%%%
%Let $G=(N, M_0, \mathcal{M_F})$ denote a Petri marked net (plant) with initial marking $M_0$ and a set of final markings $\mathcal{M_F}$.
%$\mathcal{M_F}$ can be either given by explicitly listing all its members, or characterized by a \textit{generalized mutual exclusion constraint} (GMEC)\cite{c25}. A GMEC is a pair $(\textbf{w}, k)$, where $\textbf{w}\in \mathbb{N}^m$ and $k\in \mathbb{N}$, that defines a set of markings
%\begin{center}
%$\mathcal{L}_{(\textbf{w},k)}=\{M\in\mathbb{N}^m| \textbf{w}^T\cdot M \leq k\}.$
%\end{center}
%Hereinafter, we adopt the GMEC-based representation to characterize $\mathcal{M_F}$ in $G$, i.e., let $\mathcal{M_F}=\mathcal{L}_{(\textbf{w},k)}$.
%%%%%Original Characterization of M_F%%%%%
%%%%%%%%%%%%%%%%%%%%%%%%%%%%%%%%%%%%%%%%%%
%%%%%%%%%%%%%%%%%%%%%%%%%%%%%%%%%%%%%%%%%%

%%%%%New Characterization about M_F%%%%%%%
Let $G=(N, M_0, \mathcal{M_F})$ denote a \textit{plant} consisting of a marked net and a finite set of final markings $\mathcal{M_F}\subseteq R(N, M_0)$.
Normally, set $\mathcal{M_F}$ can be given by explicitly listing all its elements.
As more general forms, set $\mathcal{M_F}$ can be characterized by linear constraints, e.g., \textit{generalized mutual exclusion constraints} (GMECs)\cite{c25}.
%\footnote{\blue Note that set $\mathcal{M_F}$ can also be described as more general forms, e.g., \textit{generalized mutual exclusion constraints} (GMECs)\cite{c25}.}.}
A GMEC is a pair $(\textbf{w}, k)$, where $\textbf{w}\in \mathbb{Z}^m$ and $k\in \mathbb{Z}$ ($\mathbb{Z}$ is the set of integers), that defines a set of markings $\mathcal{L}_{(\textbf{w},k)}=\{M\in\mathbb{N}^m| \textbf{w}^T\cdot M \leq k\}.$
%{\red Note the finite set of arbitrary final markings can also be characterized by the union of $r$ GMECs $\bigcup\limits_{i\in\{1, 2, \ldots, r\}}\mathcal{L}_{(\textbf{w}_i,k_i)}$ where $r\in \mathbb{N}^{+}$, $\textbf{w}_i\in \mathbb{N}^m$ and $k_i\in \mathbb{N}$.}
Hereinafter, we adopt the GMEC-based representation to characterize $\mathcal{M_F}$ in $G$, i.e., let $\mathcal{M_F}=\mathcal{L}_{(\textbf{w},k)}$.

%$\mathcal{M_F}$ can be either given by explicitly listing all its members, or characterized by the union of several \textit{generalized mutual exclusion constraint} (GMEC)\cite{c25}.
%A GMEC is a pair $(\textbf{w}, k)$, {\red where $\textbf{w}\in \mathbb{Z}^m$ and $k\in \mathbb{Z}$ ($\mathbb{Z}$ is the set of integers)}, that defines a set of markings $\mathcal{L}_{(\textbf{w},k)}=\{M\in\mathbb{N}^m| \textbf{w}^T\cdot M \leq k\}.$

%%%%%New Characterization about M_F%%%%%%%

%
%\subsection{Generalized Mutual Exclusion Constraints (GMECs){\rm \cite{c25}}}

%\begin{definition}
%{\rm A GMEC is a pair $(\textbf{w}, k)$, where $\textbf{w}\in \mathbb{N}^m$ and $k\in \mathbb{N}$, that defines a set of markings
%
%\begin{center}
%$\mathcal{L}_{(\textbf{w},k)}=\{M\in\mathbb{N}^m| \textbf{w}^T\cdot M \leq k\}.$
%\end{center}$\hfill\square$}
%\end{definition}

\begin{definition}\label{NB}
{\rm A marking $M\in R(N, M_0)$ of a plant $G=(N, M_0, \mathcal{M_F})$ is said to be \emph{blocking} if no final marking is reachable from it, i.e., $R(N, M)\cap \mathcal{M_F} =\emptyset$; otherwise $M$ is said to be \textit{nonblocking}.
System $G$ is \textit{nonblocking} if no reachable marking is blocking; otherwise $G$ is \textit{blocking}.$\hfill\square$}
\end{definition}

\subsection{Basis Reachability Graph (BRG){\rm \cite{c3,c8,Basis}}}
%Consider an acyclic Petri net $N$ with incidence matrix $C$, then for all $M\in \mathbb{N}^{m}$ the following holds \cite{c8}:
%\begin{equation}
%\begin{aligned}
%M+C\cdot \textbf{y}\geq \textbf{0}\ (\textbf{y}\in \mathbb{N}^{n})&\Rightarrow\\
%\exists \sigma\in \varphi^{-1}(\textbf{y})&: M[\sigma \rangle \bar{M},\ \bar{M}=M+C\cdot \textbf{y}
%\end{aligned}
%\end{equation}

\begin{definition}
{\rm Given a Petri net $N = (P, T, Pre, Post)$, transition set $T$ can be partitioned into $T=T_E\cup T_I$, where the disjoint sets $T_E$ and $T_I$ are called the \textit{explicit} transition set and the \textit{implicit} transition set, respectively.
A pair $\pi=(T_E, T_I)$ is called a \textit{basis partition} of $T$ if the $T_I$-induced sub-net of $N$ is acyclic. We denote $|T_E|=n_E$ and $|T_I|=n_I$.
Let $C_I$ be the incidence matrix of the $T_I$-induced sub-net of $N$. $\hfill\square$}
\end{definition}

Note that in a BRG with respect to a basis partition $(T_E, T_I)$, the firing information of explicit transitions in $T_E$ is explicitly encoded in the BRG, while the firing information of implicit transitions in $T_I$ is abstracted as firing vectors. Note that the selection of $T_E$ and $T_I$ does not related to the physical meaning of the transitions: the only restriction is that the $T_I$-induced sub-net is acyclic.

\begin{definition}
{\rm Given a Petri net $N = (P, T, Pre, Post)$, a basis partition $\pi=(T_E, T_I)$, a marking $M$, and a transition $t\in T_E$, we define
\begin{center}
$\Sigma(M, t)=\{\sigma\in T_{I}^{\ast}| M[\sigma\rangle M^{\prime}, M^{\prime}\geq Pre (\cdot, t)\}$
\end{center}
as the set of \emph{explanations} of $t$ at $M$, and we define

\begin{center}
$Y(M, t)=\{\varphi(\sigma)\in \mathbb{N}^{n_I}| \sigma\in \Sigma(M, t)\}$
\end{center}
as the set of \emph{explanation vectors}; meanwhile we define
{$$\Sigma_{{\rm min}}(M, t)=\{\sigma\in \Sigma(M, t)| \nexists \sigma^{\prime}\in \Sigma(M, t): \varphi(\sigma^{\prime})\lneq \varphi(\sigma)\}$$}
as the set of \emph{minimal explanations} of $t$ at $M$, and we define

\begin{center}
$Y_{{\rm min}}(M, t)=\{\varphi(\sigma)\in \mathbb{N}^{n_{I}}| \sigma\in \Sigma_{{\rm min}}(M, t)\}$
\end{center}
as the corresponding set of \emph{minimal explanation vectors}.$\hfill\square$}
\end{definition}

\begin{definition}\label{DefX}
{\rm Given a marked net $(N, M_0)$ and a basis partition $\pi=(T_E, T_I)$, its \textit{basis marking set} $\mathcal{M_B}$ is the smallest subset of reachable markings such that:
\begin{itemize}
  \item $M_0\in \mathcal{M_B}$;
  \item If $M\in \mathcal{M_B}$, then for all $t\in T_E$, for all $\textbf{y}\in Y_{\rm min}(M, t)$, $M^{\prime}=M+C_I\cdot \textbf{y}+C(\cdot, t)\Rightarrow M^{\prime}\in \mathcal{M_B}.\hfill\square$
\end{itemize}}
\end{definition}
A marking $M$ in $\mathcal{M_B}$ is called a \textit{basis marking} of $(N, M_0)$ with respect to $\pi=(T_E, T_I)$.

%%%%%Original Definition 5: BRG
%\begin{definition}
%{\rm Given a bounded net $N=(P, T, Pre, Post)$ with an initial marking $M_0$ and a basis partition $\pi=(T_E, T_I)$, its \emph{basis reachability graph} is a non-deterministic finite state automaton $\mathcal{B}$ output by Algorithm 2 in \cite{c3}. The BRG $\mathcal{B}$ is a quadruple $(\mathcal{M_B}, {\rm Tr}, \Delta, M_0)$, where
%\begin{itemize}
%  \item the state set $\mathcal{M_B}$ is the set of basis markings;
%  \item the event set ${\rm Tr}$ is the set of pairs $(t, \textbf{y})\in T_E\times \mathbb{N}^{n_{I}}$;
%  \item the transition relation $\Delta=\{(M_1, (t, \textbf{y}), M_2)| t\in T_E, \textbf{y}\in Y_{\rm min}(M_1, t), M_2=M_1+C_I\cdot \textbf{y}+C(\cdot, t)\}$;
%  \item the initial state is the initial marking $M_0$.$\hfill\square$
%\end{itemize}}
%\end{definition}
%%%%%%%%%%%%%%%%%%%%%%%%%%%%%%%%

\begin{definition}
{\rm Given a bounded marked net $\langle N, M_0\rangle$ and a basis partition $\pi=(T_E, T_I)$, its \emph{basis reachability graph} is a deterministic finite state automaton $\mathcal{B}=(\mathcal{M_B}, {\rm Tr}, \Delta, M_0)$, where the state set $\mathcal{M_B}$ is the set of basis markings, the event set ${\rm Tr}$ is the finite set of pairs $(t, \textbf{y})\in T_E\times \mathbb{N}^{n_{I}}$, the transition relation $\Delta=\{(M_1, (t, \textbf{y}), M_2)\mid t\in T_E, \textbf{y}\in Y_{\rm min}(M_1, t), M_2=M_1+C_I\cdot \textbf{y}+C(\cdot, t)\}$ and the initial state is the initial marking $M_0$.$\hfill\square$}
\end{definition}

We extend in the usual way the definition of transition relation to consider a sequence of pairs $\sigma\in {\rm Tr}^*$ and write $(M_1,\sigma,M_2)\in\Delta$ to denote that from $M_1$ sequence $\sigma$ yields $M_2$.
%\end{lemma}
%\begin{proof}
%Since $N$ is acyclic, the state equation provides fully representation its marking reachability.
%Therefore if $M+C\cdot \textbf{y}\geq \textbf{0}$ we conclude that from $M$, among all firing sequences whose firing vector mapped from function $\varphi$ as $\textbf{y}$, there exists at least a firing sequence $\sigma\in \varphi^{-1}(\textbf{y})$ such that $M[\sigma \rangle \bar{M}$ where $\bar{M}=M+C\cdot \textbf{y}$.
%\end{proof}

\begin{definition}\label{DefIR}
{\rm Given a marked net $\langle N, M_0\rangle$, a basis partition $\pi=(T_E, T_I)$, and a basis marking $M_b\in \mathcal{M_B}$, we define
$R_I(M_b)=\{M\in \mathbb{N}^m\mid(\exists \sigma\in T_I^\ast)\; M_b[\sigma\rangle M\}$
as the \emph{implicit reach} of $M_b$.$\hfill\square$}
%, and we define
%\begin{center}
%$R_I^{\rm max}(M_b)=\{M\in \mathbb{N}^m|(\exists \sigma\in T_I^\ast) (\nexists \sigma^{'}\in T_I^\ast: \textbf{y}_{\sigma^{'}}\gneqq \textbf{y}_{\sigma})\; M_b[\sigma\rangle M\}$
%\end{center}
%as the \emph{maximal implicit reach} of $M_b$.
\end{definition}

The implicit reach of a basis marking $M_b$ is the set of all markings that can be reached from $M_b$ by firing only implicit transitions.
Since the $T_I$-induced sub-net is acyclic, by Proposition \ref{ProX}, it holds that:
\begin{center}
$R_I(M_b)=\{M\in\mathbb{N}^m\mid (\exists\textbf{y}_I\in\mathbb{N}^{n_I})\; M=M_b+C_I\cdot \textbf{y}_I\}.$
\end{center}

\section{BRG and Nonblockingness Verification}\label{sec2}
%The efficient verification of nonblockingness in Petri nets without an exhaustive enumeration of the state space remains an open issue.
%To attempt to discover a solution to the nonblockingness verification problem by using the BRG-based method, in \cite{Gu}, we first define the set of \textit{i-coreachable markings} and introduce the notion of \textit{unobstructiveness} of a BRG.

To efficiently solve the NB-V problem in Petri nets without constructing the reachability graph, we attempted to use the BRG-based approach in \cite{Gu}.
%{\red The work in \cite{Gu} discusses the {\red NB-V problem} by leveraging the BRG.
%{\blue If a plant is verified to be nonblocking, all nodes included in its corresponding BRG are nonblocking.}
%However, the converse may not be valid, i.e., {\blue the fact that all nodes in the BRG of a plant are nonblocking does not necessarily imply that the plant is nonblocking}. To help clarify, an example is provided in the following.}
However, as observed in \cite{Gu}, the classical BRG does not necessarily encode all information needed to test nonblockingness. To help clarify, an example is provided in the following.

\begin{example}\label{E1}
%{\rm Consider a parameterized plant $G=(N, M_0, \mathcal{M_F})$ in Fig. \ref{Fig1} with $M_0=[2\ 0\ 1]^{\rm T}$ and $\mathcal{M_F}=\{M_0\}$. In this net, $Pre(p_2, t_3)=\alpha$ is set to be a parameter ($\alpha\in \mathbb{N}$). Assuming {\blue $T_E = \{t_2\}$}, the BRG of this net (regardless of the value of $\alpha$) is also shown in the same figure, {\red where $M_{b0}=M_0, M_{b1}=[1\ 1\ 1]^{\rm T}, M_{b2}=[0\ 2\ 1]^{\rm T}, \textbf{y}_1=\textbf{y}_2=[0\ 0]^{\rm T}$ and $\textbf{y}_3=[1\ 0]^{\rm T}$.} The reachability graphs for $\alpha=1$ and $\alpha=2$ are shown in Fig. \ref{Fig2}.}
Consider a parameterized plant $G=(N, M_0, \mathcal{M_F})$ in Fig. \ref{Fig1} with $M_0=[2\ 0\ 1]^{\rm T}$ and $\mathcal{M_F}=\{M_0\}$. In this net, $Pre(p_2, t_3)=\alpha$ is set to be a parameter ($\alpha\in \mathbb{N}$). Assuming that $T_E = \{t_2\}$, the BRG of this net (regardless of the value of $\alpha$) is also shown in the same figure, where $M_{b0}=M_0$ and $Y_{\rm min}(M_{b0}, t_2) = \{\textbf{y}_1\}$.
It can be computed that $\textbf{y}_1=[1\ 0]^{\rm T}$, which implies that the firing of sequence $\sigma = t_1$ is the prerequisite (the minimal one) of the firing of explicit transition $t_2$ at marking $M_{b0}$. The reachability graphs for $\alpha=1$ and $\alpha=2$ are shown in Fig. \ref{Fig2}.

\begin{figure}[t]
\begin{center}
\begin{minipage}[t]{0.6\textwidth}
\includegraphics[width=4cm]{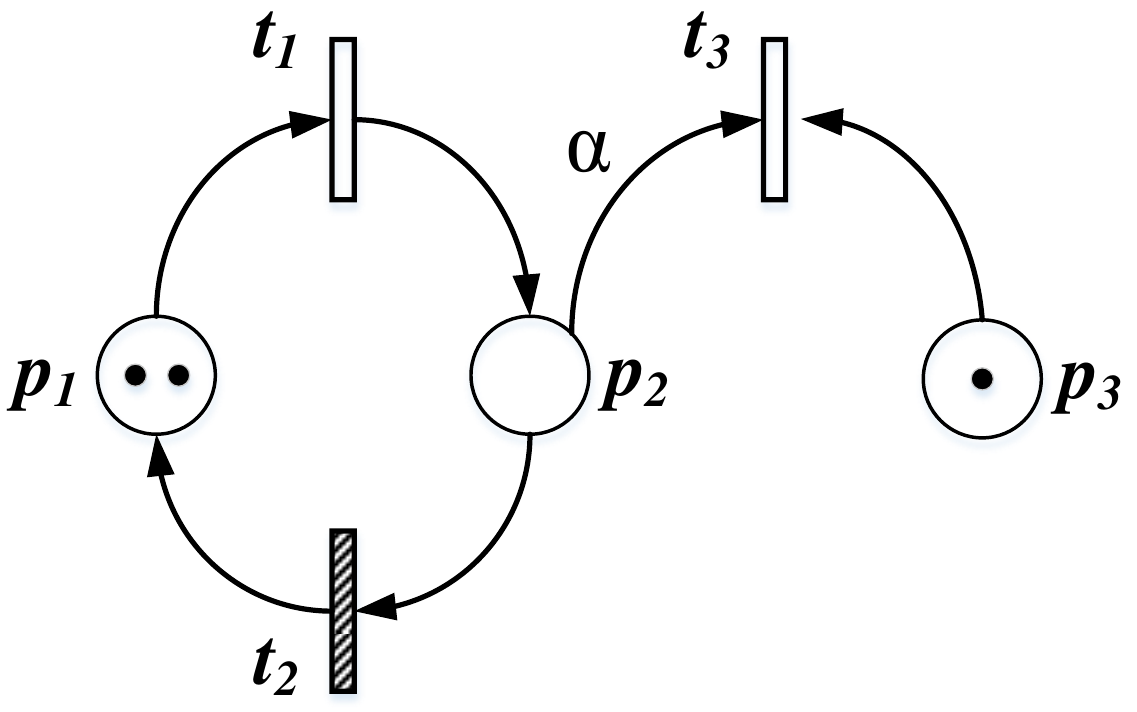}
\end{minipage}
\begin{minipage}[t]{0.2\textwidth}
\includegraphics[width=1.1cm]{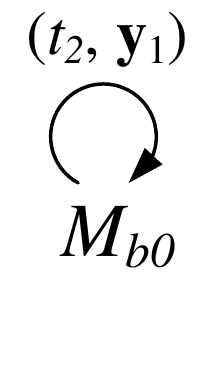}
\end{minipage}
\caption{A parameterized plant $G$ with $T_E=\{t_2\}$ marked with shadow (left) and its BRG $\mathcal{B}$ (right).}\label{Fig1}
\end{center}
\end{figure}
%\begin{figure}[H]
%\begin{center}
%\begin{minipage}[t]{0.6\textwidth}
%\includegraphics[width=4cm]{new1107-eps-converted-to.pdf}
%\end{minipage}
%\begin{minipage}[t]{0.2\textwidth}
%\includegraphics[width=2.3cm]{BRG_0830.pdf}
%\end{minipage}
%\caption{A parameterized plant $G$ with $T_E=\{t_1\}$ marked with shadow (left) and {\blue its BRG $\mathcal{B}$ (right).}}\label{Fig1}
%\end{center}
%\end{figure}

\begin{figure}[t]
\begin{center}
\begin{minipage}[t]{0.535\textwidth}
\includegraphics[width=4cm]{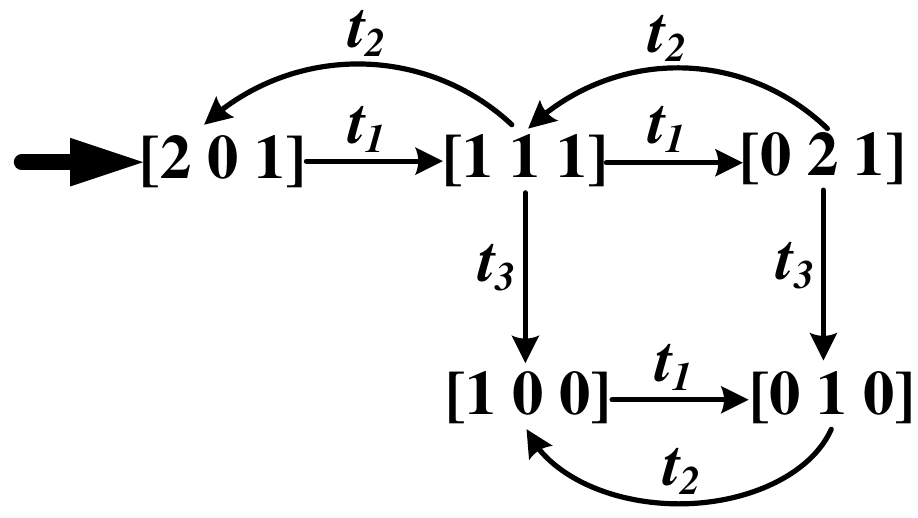}
\end{minipage}
\begin{minipage}[t]{0.44\textwidth}
\includegraphics[width=4cm]{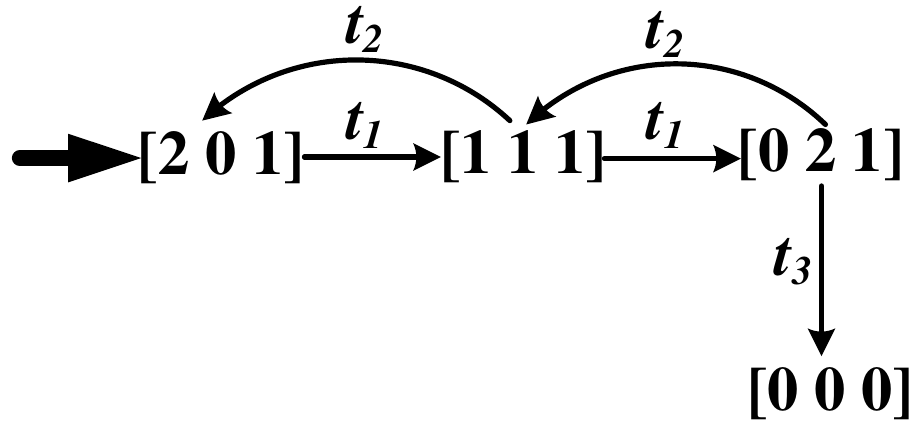}
\end{minipage}
\caption{Reachability graph of $G$ in Fig. \ref{Fig1} with $\alpha=1$ (left) and $\alpha=2$ (right).}\label{Fig2}
\end{center}
\end{figure}
\end{example}

%It can be concluded that {\blue the only basis marking $M_{b0}$ is nonblocking in the BRG in Figure \ref{Fig1} (right)}.
By inspection of the two reachability graphs, one can verify that $G$ is deadlock-free if $\alpha=1$ and not deadlock-free if $\alpha=2$.
When $\alpha=1$ $G$ is blocking due to the livelock composed by two markings $[1\ 0\ 0]^{\rm T}$ and $[0\ 1\ 0]^{\rm T}$.
When $\alpha=2$ $G$ is also blocking because of the non-final deadlock $[0\ 0\ 0]^{\rm T}$.
However, these blocking conditions are not captured in the BRG which, in both cases, consists of a unique node $M_{b0} = M_0$ which is also final.$\hfill\square$

Example \ref{E1} shows that when all basis markings in the BRG are nonblocking, this does not necessarily imply that all reachable markings in the corresponding plant are nonblocking.
%We attempted to discover a solution of this problem in {\red [CDC paper]} by first defining a new property of the BRG, i.e., \emph{unobstructiveness}, and show that a Petri net is blocking if its BRG is obstructed. However, the converse is not true in that the possible presences of \textit{livelocks} or \textit{deadlocks} among non-basis markings:
Specifically, as we mentioned in Section \ref{Section1}, two types of blocking markings should be analyzed, i.e., those are dead but non-final, and those are included in livelocks (ergodic strongly-connected components of the reachability graph containing non-final non-dead markings).
%%%%%%%%%%Original Description%%%%%%%%%%%%
%\begin{enumerate}
%  \item [$-$] dead but non-final;
%  \item [$-$] livelocks, i.e., ergodic strongly-connected components of non-dead markings.
%%  Once a marking in a livelock is reached the future evolution will remain within this component.
%\end{enumerate}
%%%%%%%%%%Original Description%%%%%%%%%%%%

%Notice that the occurrence of such livelock and deadlock problems stems from the abstraction of information inherent in the basis marking approach, and the unobstructiveness of a BRG may not completely characterize the nonblockingness of the Petri net.
Notice that when tackling the NB-V problem by using the basis marking approach, there may exist some (i) dead and non-final markings, and/or (ii) livelock markings that are not basis. Since such markings do not belong to set $\mathcal{M_B}$, they are not shown in the corresponding BRG.
Therefore, the classical structure of BRGs needs to be revised to encode additional information for checking nonblockingness.
To this end, in the following, we propose a structure namely minimax-BRG and show how it can be leveraged on solving the NB-V problem.

\section{Minimax Basis Markings and Minimax-BRGs}\label{MinimaxBRG}
%\section{Verifying Nonblockingness of Deadlock-Free Petri Nets Using Minimax-BRGs}\label{NewSection}

%{\red In this section, we introduce the minimax-BRG and characterize its properties relevant for NB-V.}
%A novel structure namely the \textit{minimax-basis reachability graph} (minimax-BRG) is defined, and its properties are studied.
%Moreover, we extend the notions of i-coreachable marking set and unobstructiveness \cite{Gu} with regard to the minimax-BRG.
%Using the minimax-BRG instead of the reachability graph alleviates the state explosion conundrum compared with the exhaustive enumeration of the state space.
%To construct a minimax-BRG is often more efficient than that of a corresponding expanded BRG, while both of these two semi-structure-based methodologies can resolve the potential livelock problems with respect to nonblockingness verification.

\subsection{Minimax Basis Markings}\label{SRGchapter}
%\subsection{Maximal Explanations and Minimax Basis Markings}\label{SRGchapter}
%In order to solve the nonblockingness verification problem, we proposed an algorithm to obtain a structure which we call the minimax BRG that contains not only all basis markings but also some considered non-basis markings.
%Although the expanded BRG can be used on solving the livelock problem, its efficiency needs to be further improved.
%Since all markings that reached by firing a feasible transition sequence that end with an explicit one are exhaustively collected, which will inevitably lead to a redundancy, i.e., not all of them are related to the possibly livelock components.
%Similar with BRG, it considerably improves the computational efficiency comparing with the expanded BRG. In detail, less additional non-basis markings and transition relations are computed so that the redundancy can be alleviated significantly.
To define the minimax-BRG, we first introduce the set of \textit{minimax basis markings}.
As two prerequisite concepts, we define \textit{maximal explanations} and \textit{maximal explanation vectors} as follows.

\begin{definition}\label{DEFMAX}
{\rm Given a Petri net $N = (P, T, Pre, Post)$, a basis partition $\pi=(T_E, T_I)$, a marking $M$, and a transition $t\in T_E$, we define
%\begin{center}
%$\Sigma(M, t)=\{\sigma\in T_{I}^{\ast}| M[\sigma\rangle M^{\prime}, M^{\prime}\geq Pre (\cdot, t)\}$
%\end{center}
%as the set of \emph{explanations} of $t$ at $M$, and
%\begin{center}
%$Y(M, t)=\{\varphi(\sigma)\in \mathbb{N}^{n_I}| \sigma\in \Sigma(M, t)\}$
%\end{center}
%as the set of \emph{explanation vectors}; meanwhile we define
\begin{center}
$\Sigma_{{\rm max}}(M, t)=\{\sigma\in \Sigma(M, t)| \nexists \sigma^{\prime}\in \Sigma(M, t): \varphi(\sigma^{\prime})\gneq \varphi(\sigma)\}$
\end{center}
as the set of maximal explanations of $t$ at $M$, and

\begin{center}
$Y_{{\rm max}}(M, t)=\{\varphi(\sigma)\in \mathbb{N}^{n_{I}}| \sigma\in \Sigma_{{\rm max}}(M, t)\}$
\end{center}
as the corresponding set of maximal explanation vectors.$\hfill\square$}
\end{definition}

%\begin{remark}
%Given a marked net $\langle N, M_0\rangle$ with $\pi=(T_E, T_I)$, for all $M\in R(N, M_0)$ and for all $t\in T_E$, $Y_{\rm max}(M, t)$ is the set of maximal elements in poset $Y(M, t)$.$\hfill\square$
%\end{remark}
From the standpoint of \textit{partial order set} (poset), the set of maximal explanation vectors $Y_{\rm max}(M,t)$ is the set of \textit{maximal elements} in the corresponding poset $Y(M,t)$.
Note that, as is the case for the set of minimal explanation vectors $Y_{{\rm min}}(M, t)$\cite{c3,c8,Basis}, $Y_{{\rm max}}(M, t)$ may not be a singleton.
In fact, there may exist multiple maximal firing sequences $\sigma_I\in T_I^{*}$ that enable an explicit transition $t$.
Next, we define \textit{minimax basis markings} in an iterative way as follows.
\begin{definition}\label{Def1}
{\rm Given a marked net $\langle N, M_0\rangle$ with a basis partition $\pi=(T_E, T_I)$, its minimax basis marking set $\mathcal{M_{B_M}}$ is recursively defined as follows
\begin{enumerate}
  \item [(a)] $M_0\in \mathcal{M_{B_M}}$;
  \item [(b)] $M\in \mathcal{M_{B_M}}$, $t\in T_E, \textbf{y}\in Y_{\rm min}(M, t)\cup Y_{\rm max}(M, t)$, $M^{\prime}=M+C_I\cdot \textbf{y}+C(\cdot, t)\Rightarrow M^{\prime}\in\mathcal{M_{B_M}}$.
%  \item If $M\in \mathcal{M_{B_M}}$, $\forall t\in T_E, \forall \textbf{y}^{\prime}\in Y_{\rm max}(M, t)$, $(M^{\prime\prime}=M+C_I\cdot \textbf{y}^{\prime}+C(\cdot, t))\Rightarrow M^{\prime\prime}\in\mathcal{M_{B_M}}$.
\end{enumerate}
A marking in $\mathcal{M_{B_M}}$ is called a minimax basis marking of the marked net with $\pi=(T_E, T_I)$.$\hfill\square$}
\end{definition}

In practice, the set of minimax basis markings is a smaller subset of reachable markings that contains the initial marking and is closed by reachability through a sequence that contains an explicit transition and one of its maximal or minimal explanations.
Meanwhile, note that for a bounded marked net, $\mathcal{M_B}\subseteq\mathcal{M_{B_M}}$ holds.
To compute $Y_{\rm min}(M, t)$, one may refer to Algorithm 1 in \cite{c3}.
We introduce in Algorithm \ref{AlgoMax} how to calculate $Y_{\rm max}(M, t)$ for a given marking $M$ and an explicit transition $t$. The basic idea is first to iteratively enumerate all explanation vectors in $Y(M, t)$ (not necessarily stored), and then collect the set of maximal elements in $Y(M, t)$.

\begin{algorithm}[t]
\caption{Calculation of $Y_{\rm max}(M,t)$} %算法的名字
\begin{algorithmic}[1]
\REQUIRE A marked net $\langle N, M_0\rangle$, a basis partition $\pi= (T_E, T_I )$, a marking $M\in R(N, M_0)$, and $t\in T_E$
\ENSURE $Y_{\rm max}(M,t)$
\STATE {$\Gamma:=\left[\begin{array}{c|c}
                       C_I^{\rm T} & I_{n_{I}\times n_{I}} \\ \hline
                       A & B\\\end{array}\right]$
where $A:=(M-Pre(\cdot, t))^{\rm T}$ and $B:=\textbf{0}_{n_I}^{\rm T}$};
\STATE {Subroutine: update $\Gamma$ through lines 2$\--$12 in Algorithm 1 of \cite{c3};}
%%%%%%%%%%%%%%%%%%%%%%%%%%%%%以下为新添加的代码,用于计算Y(M, t)以及Ymax(M, t)%%%%%%%%%%%%%%%%%%
\STATE {$\alpha:={\rm row\_size}(\Gamma)$, $\alpha_{\rm old}:=0$}, and $\alpha^{\prime}_{\rm old}:=n_I$;
\WHILE {$\alpha_{\rm old}-\alpha\neq0$}
\STATE {$\alpha_{\rm old}:=\alpha$};
\FOR{$k=1:n_{I}$,}  % If 语句，需要和EndIf对应
\FOR {$l=(\alpha^{\prime}_{\rm old}+1): \alpha_{\rm old}$,}
\STATE {$R:=[\Gamma(l, \cdot)+\Gamma(k, \cdot)]$;}
\IF {$R\geq \textbf{0}$ and $\nexists \Gamma(i, \cdot) = R$ where $i\in\{(n_I+1), \cdots, \alpha_{\rm old}\}$,}
\STATE {$\Gamma_{\rm new}:= \left[
           \begin{array}{cc}
           \Gamma\\ \hline R
           \end{array}
           \right];$}
\ENDIF
\ENDFOR
\ENDFOR
\STATE {$\alpha:={\rm row\_size}(\Gamma_{\rm new})$, $\alpha^{\prime}_{\rm old}:= \alpha_{\rm old}$, and $\Gamma:=\Gamma_{\rm new}$};
\ENDWHILE
%\STATE {\blue Delete $\Gamma(n_I+1, \cdot)$ if $\Gamma(n_I+1, \cdot)\ngeq \textbf{0}$;}
\STATE Let $Y(M, t)$ be the set of row vectors in the updated sub-matrix $B=\Gamma((n_I+1): \alpha, (m+1): (m+n_I))$;
\STATE Let $Y_{\rm max}(M,t)$ be the set of maximal elements in $Y(M, t)$.
\end{algorithmic}\label{AlgoMax}
\end{algorithm}
The computation as to $Y(M, t)$ is presented through lines 1$\--$16.
To put it simply, as a \textit{breadth-first-search} technique, all possible firing vectors $\textbf{y}\in \mathbb{N}^{n_I}$ such that $\sigma\in\varphi^{-1}(\textbf{y})$ is an explanation of $t$ at $M$ (i.e., $M[\sigma\rangle M^{\prime}[t\rangle$) are iteratively searched and enumerated. A detailed description is shown as follows.
%Note that if $Y(M, t)\neq \emptyset$, the \textit{subroutine} in line 2 (lines 2$\--$12 in Algorithm 1 in \cite{c3}) delivers a matrix

Initially, at line 1, the row $A =(M-Pre(\cdot, t))^{\rm T}$ is either nonnegative or contains at least a negative element.
The former implies that $t$ is sufficiently enabled at $M$ (thus $\textbf{0}_{n_I}\in Y(M, t)$). The latter suggests that the number of tokens in the corresponding place(s) as to $M$ is insufficient.
Then, at line 2, we call a subroutine that consists of lines 2$\--$12 in Algorithm 1 of \cite{c3} to process the matrix $\Gamma$.
Precisely, this procedure enumerates part of the explanation vectors (not all) by iteratively updating $\Gamma$, i.e., adding some considered rows in $\left[
           \begin{array}{c|c}
           C_I^{\rm T} & I_{n_I\times n_I}\\
           \end{array}
           \right]$
to rows in
$\left[
           \begin{array}{c|c}
           A & B\\
           \end{array}
           \right]$ that contain negative elements to neutralize them eventually.
%The \textbf{physical meaning} of this manipulation is to test if the firing of some implicit transitions $t_{1^*}, t_{2^*}, \ldots, t_{k^*}\in T_I$ ($t_{i^*}$ corresponds to the $i^*$th$\ (i\in\{1,2,\ldots, k\})$ row of the matrix $C^T_I$) is the prerequisite of firing the explicit transition $t\in T_E$ at $M$.
The physical interpretation of this manipulation is to test whether the firing of an implicit transition $t_{i^*}\in T_I$ ($t_{i^*}$ corresponds to the $i^*$th row of the matrix $C^T_I$) at marking $M$ can replenish the corresponding token counts in $M(p)$ such that $M^{\prime}(p)-Pre(p, t)\geq \textbf{0}$ where $M[t_{i^*}\rangle M^{\prime}$.
As a result, if $Y(M, t)\neq \emptyset$, the sub-matrix $\left[
           \begin{array}{c|c}
           A & B\\
           \end{array}
           \right]$
contains all nonnegative rows and the corresponding explanation vectors are stored individually in the form of row vectors in sub-matrix $B$.
%which implies that $M + C_I\cdot (\sum^k _{i = 1}I_{n_I\times n_I}(i^*, \cdot))^T - Pre(\cdot, t)\geq \textbf{0}$

To complete $Y(M, t)$, analogously, from lines 6$\--$12, we add each of the rows in $\left[
           \begin{array}{c|c}
           C_I^{\rm T} & I_{n_I\times n_I}\\
           \end{array}
           \right]$
to rows in
$\left[
           \begin{array}{c|c}
           A & B\\
           \end{array}
           \right]$ in the updated $\Gamma$.
%{\blue The \textbf{physical interpretation} of this manipulation is to test whether each implicit transition $t_{i^*}\in T_I$ ($t_{i^*}$ corresponds to the $i^*$th row of the matrix $C^T_I$) at a marking $M$ with $M[\sigma t\rangle$ ($\sigma\in Y(M, t)$) can also validate $M[\sigma t_{i^*}t\rangle$: if so, a new explanation $\sigma t_{i^*}$ will be derived.}
If the obtained new row, e.g., $R = [C_I^{\rm T}(i^*, \cdot) + A(j^*, \cdot) \mid I_{n_I\times n_I}(i^*, \cdot) + B(j^*, \cdot)]$, is nonnegative and does not equal to any of the rows in $\left[
           \begin{array}{c|c}
           A & B\\
           \end{array}
           \right]$, it is then recorded in $\left[
           \begin{array}{c|c}
           A & B\\
           \end{array}
           \right]$ as a new extended row and matrix $\Gamma$ will be updated.
This act implies that $M + C_I\cdot (I_{n_I\times n_I}(i^*, \cdot) + B(j^*, \cdot))^T - Pre(\cdot, t)\geq \textbf{0}$. Thus, it is deduced that the vector $(I_{n_I\times n_I}(i^*, \cdot) + B(j^*, \cdot))^T$ is another explanation vector of $t$ at $M$ and it will be recorded in the sub-matrix $B$.
     %and $\Gamma$ will be updated, {\blue \textbf{which implies that $M + C_I\cdot (B(i^*, \cdot))^T - Pre(\cdot, t)\geq \textbf{0}$ and the vector $(B(i^*, \cdot))^T$ induced by $R$ is unique.}}
%{\blue In fact, each vector $(B(i^*, \cdot))^T$ mentioned above encodes an explanation vector information, since $(R(1:m))^T = M + C_I \cdot (R(m+1:m+n_I))^T - Pre(\cdot, t)\geq \textbf{0}$ implies that $(R(m+1:m+n_I))^T$ is an explanation vector of $t$ at $M$.
%{\blue In fact, the vector $(B+I_{n_I\times n_I}(i^*, \cdot))^T$ mentioned above induced by $R$ is an explanation vector of $t$ at $M$.}

Iteratively, represented in the sub-matrix $B$ of the updated matrix $\Gamma_{\rm new}$, all explanations of $M$ at $t$ can be collected.
The computation of $Y(M, t)$ ends when the sub-matrix $\left[
           \begin{array}{c|c}
           A & B\\
           \end{array}
           \right]$ of $\Gamma_{\rm new}$ reaches a fixed point. Note that due to the boundness of the net and the acyclicity of the $T_I$-induced sub-net, Algorithm \ref{AlgoMax} will not run endlessly, since $Y(M, t)$ is not infinite.
%$Y(M,t)$ is represented as the set of row vectors in the updated sub-matrix $B$ of $\Gamma_{\rm new}$.
Finally, at line 17, the set of maximal explanations is obtained by collecting all the maximal rows in $Y(M, t)$.
%
%In fact, a new explanation vector $\textbf{y}^{\prime}\in \mathbb{N}^{n_I}$ of $t$ at $M$ can be collected based on $R$, since there exists a firing sequence $\sigma^{\prime}\in\varphi^{-1}(\textbf{y}^{\prime})$ such that $M[\sigma^{\prime}\rangle M^{\prime\prime}[t\rangle$.

%Iteratively, represented in the sub-matrix $B$ of the updated matrix $\Gamma_{\rm new}$, all explanations of $M$ at $t$ can be computed.
%Stage 2 ends when $\alpha$ equals to $\alpha_{\rm old}$, meaning that sub-matrix $\left[
%           \begin{array}{c|c}
%           A & B\\
%           \end{array}
%           \right]$ reaches a fixed point.
%Finally, the set of all explanation vectors in $Y(M, t)$ can be generated.
%By removing each of the rows in $B$ that can be covered by another row, the set of maximal explanation vectors $Y_{\rm max}(M, t)$ can be established.

%{\blue \textbf We analyze the complexity of Algorithm \ref{AlgoMax}. At stage 1, the complexity for each of the iteration in the \textit{while} loop is $\mathcal{O}(|P|\cdot n_I)$. Since the marked net is bounded and $T_I$-induced sub-net is acyclic, the \textit{while} loop shall be terminated within a finite number of times. However, since different systems may have different structures, the above iteration counts cannot be expressed as a simple function of the system structure. Therefore, the complexity of Algorithm 1 cannot be quantified generally.}

\begin{remark}
We analyze the complexity of Algorithm \ref{AlgoMax}. The complexity of the \textit{while} loop (lines 3$\--$14) can be estimated as follows.
%we have $|Y(M, t)|$ iterations from lines 3 and 6, and 8 (part of line 8 that determines if $\nexists \Gamma(i, \cdot) = R$); $n_I$ iterations from lines 5 and 8 (part of line 8 that tests $R\neq Γ(i, ・)$ is linear in $n_I+nI = 2n_I$). Thus, the worst-case time complexity of this loop is $\mathcal{O}(|Y(M, t)|^3n^2_I)$.
we have $|Y(M, t)|$ iterations from lines 3, 6 and 8 (determine if $R$ is unique) and $n_I$ iterations from lines 5. Thus, the worst-case time complexity of this loop is $\mathcal{O}(|Y(M, t)|^3n_I)$.
On the other hand, the complexity of line 17 is less than the above-mentioned loop.
Thus, the overall complexity is cubic in $|Y(M, t)|$.$\hfill\square$
\end{remark}

%{\red We analyze the complexity of Algorithm \ref{AlgoMax}. For stage 1, the complexity is $\mathcal{O}(|P|\cdot |T_I|)$; for stage 2, the complexity is $\mathcal{O}(|T_I|^3)$. Hence, the overall complexity of Algorithm 1 is $\mathcal{O}(|P|\cdot |T_I|+|T_I|^3)$, which is {\blue polynomial in the size of the net.}}

%An example illustrating how this algorithm works is provided in the Appendix.

%\begin{remark}\label{RemX}
%{\rm The set of minimax basis marking $\mathcal{M_{B_M}}$ is a superset of the set of basis markings $\mathcal{M_B}$ defined in Definition \ref{DefX}, i.e., $\mathcal{M_{B_M}}\supseteq\mathcal{M_B}$.
%In fact, $\mathcal{M_B}$ can be recursively computed as in Definition \ref{Def1} but assuming that in condition (b) $\textbf{y}\in Y_{\rm min}(M, t)$ holds, i.e., only minimal explanations are considered.}
%\end{remark}

%Note that if $|Y_{\rm max}(M, t)|=1$ then it will be the greatest element that is the unique maximal element in $Y(M, t)$. The condition of $Y_{\rm max}(M, t)$ being a singleton is discussed in Section \ref{Singleton}.

\subsection{Minimax Basis Reachability Graph}
%{\blue The minimax-BRG is defined in the following.}
\begin{definition}\label{Def2}
Given a bounded marked net $\langle N, M_0\rangle$ and a basis partition $\pi=(T_E, T_I)$, its minimax-BRG is a deterministic finite state automaton $\mathcal{B_M}=(\mathcal{M_{B_M}}, {\rm Tr_{\mathcal{M}}}, \Delta_{\mathcal{M}}, M_0)$, where $\mathcal{M_{B_M}}$ is the set of minimax basis markings, ${\rm Tr}_\mathcal{M}$ is a finite set of pairs $(t, \textbf{y})\in T_E\times \mathbb{N}^{n_{I}}$, $\Delta_{\mathcal{M}}$ is the transition relation $\{(M_1, (t, \textbf{y}), M_2)\mid t\in T_E; \textbf{y}\in (Y_{\rm min}(M_1, t)\cup Y_{\rm max}(M_1, t)), M_2=M_1+C_I\cdot \textbf{y}+C(\cdot, t)\}$ and $M_0$ is the initial marking.$\hfill\square$
\end{definition}

We extend the definition of transition relation $\Delta_{\mathcal{M}}$ for sequences of pairs $\sigma^{+}=(t_1, \textbf{y}_1), (t_2, \textbf{y}_2), \cdots, (t_k, \textbf{y}_k)\in {\rm Tr}_\mathcal{M}^*$ and write $(M_1, \sigma^{+}, M_2)\in\Delta_{\mathcal{M}}$ to denote that from $M_1$ sequence $\sigma^{+}$ yields $M_2$ in $\mathcal{B_M}$.

According to Definitions \ref{Def1} and \ref{Def2}, to build a minimax-BRG, one may refer to the construction procedure of a BRG (e.g., see Algorithm 2 in \cite{c3}).
%The difference is that the construction of minimax-BRGs requires to take not only all minimal explanation vectors but also all maximal ones into consideration.
The difference is that the construction of minimax-BRGs requires taking both minimal and maximal explanation vectors into consideration.
We briefly explain the construction procedure as follows.
First, the set $\mathcal{M_{B_M}}$ is initialized as $\{M_0\}$.
Then, for all untested markings $M\in \mathcal{M_{B_M}}$ and for all explicit transitions $t\in T_E$, it is required to check whether there exist explanation vectors $\textbf{y}\in Y_{\rm min}(M,t)\cup Y_{\rm max}(M,t)$: if exist, the corresponding minimax basis marking (i.e., $M^{\prime}=M+C_I\cdot \textbf{y}+C(\cdot, t)$) is computed and stored in $\mathcal{M_{B_M}}$ (on the condition that $M^{\prime}$ is not included in the set $\mathcal{M_{B_M}}$ before).
Moreover, the set of pairs $(t,\textbf{y})$ and transition relations between $M$ and $M^{\prime}$ are stored in ${\rm Tr_{\mathcal{M}}}$ and $\Delta_\mathcal{M}$, respectively.
Iteratively, the minimax-BRG $\mathcal{B_M}$ can be constructed.
%Algorithm \ref{Algo3} stops when there is no unchecked marking in $\mathcal{M_{B_M}}$.
We exemplify this procedure in Example \ref{EXPSRG}.

%Algorithm \ref{Algo3} computes a minimax-BRG.
%The set $\mathcal{M_{B_M}}$ is initialized at $\{M_0\}$.
%At the end of the procedure, it contains the set of minimax basis markings.
%For all untested markings $M\in \mathcal{M_{B_M}}$, i.e., those with no tag, and for all explicit transitions $t\in T_E$, we check whether there exist explanation vectors $\textbf{y}\in Y_{\rm min}(M,t)$ or $\textbf{y}\in Y_{\rm max}(M,t)$.
%If such explanation vectors exist, we compute all minimax basis markings (i.e., $M^{\prime}=M+C_I\cdot \textbf{y}+C(\cdot, t)$) and store them in $\mathcal{M_{B_M}}$.
%Moreover, the set of pairs $(t,\textbf{y})$ and transition relations between $M$ and $M^{\prime}$ are stored in ${\rm Tr_{\mathcal{M}}}$ and $\Delta_\mathcal{M}$, respectively.
%
%Algorithm \ref{Algo3} stops when there is no unchecked marking in $\mathcal{M_{B_M}}$.
%Comparing with the construction of the BRG, where one needs to compute the minimal explanation vectors\cite{c8,Basis}, Algorithm \ref{Algo3} requires to compute all markings that are reachable from the initial marking by firing not only all minimal explanation vectors but also all maximal ones.
As for the complexity of constructing the minimax-BRG, in common with the BRG, the upper bound of states in a minimax-BRG is the size of
the reachability space of a net (consider $T_E=T$ and $T_I=\emptyset$).
Nonetheless, first, the building of a minimax-BRG does not require constructing the reachability graph.
%Then, we will show in the following sections that, when tackling the NB-V problem, the minimax-BRG-based approach is general and can be directly applied to arbitrary bounded plants (the only restriction is that the $T_I$-induced sub-net is acyclic). This is a major practical advantage with respect to other abstraction approaches that are based on particular structures or symmetries, and require significant analysis of the model in a preliminary stage before they can be applied.
%Further, our numerical results (e.g., see Section \ref{sec2.3} and \cite{GitHub}) show that the minimax-BRG-based technique achieves efficiency in some considered cases.
Then, our numerical results (e.g., see Section \ref{sec2.3} and \cite{GitHub}) show that the minimax-BRG can often be more compact in size than that of the reachability graph in the considered cases.
%Accordingly, when it comes to a related problem of NB-V, i.e., \textit{nonblocking enforcement}, which consists of designing a \textit{supervisor} to ensure that the controlled plant does not reach a blocking marking, a supervisor designed based on the minimax-BRG can also be more compact than that of a reachability-graph-based one.
%Note that in this work we only address the NB-V problem. A related problem, which we do not consider here, is that of \textit{nonblockingness enforcing}, which consists in designing a \textit{supervisor} to ensure that the controlled plant does not reach a blocking state.

\begin{remark}
Similar to the BRG, note that the selection of the basis partition may change the computational efficiency of constructing the minimax-BRG.
In general, the larger is the set $T_I$ of implicit transitions in a basis partition $\pi=(T_E, T_I)$, the smaller is the number of nodes in the minimax-BRG and the time required for its construction. See Section \uppercase\expandafter{\romannumeral4} in \cite{c3} for a discussion on how to choose basis partitions.$\hfill\square$
%In general, for a basis partition $\pi=(T_E, T_I)$,  tends to deliver a minimax-BRG with fewer nodes and cost less time.
%Such selection problem has been considered and studied in the standard BRG-based approach \cite{c3}, i.e., given a Petri net $N=(P, T, Pre, Post)$, the complexity of obtaining a basis partition that contains a maximal number of the implicit transitions is $\mathcal{O}(|P|\cdot|T|^2)$.
%Although we focus on the minimax-BRG in this paper, the rationale is the same.
\end{remark}

\begin{example}\label{EXPSRG}
Consider again the parameterized plant $G=(N, M_0, \mathcal{M_F})$ in Fig. \ref{Fig1} (left) with $\alpha=1$ and $T_E=\{t_2\}$.
We briefly introduce how to construct its minimax-BRG $\mathcal{B_M}$.
According to Definition \ref{Def1}, $M_{b0} = M_0$ is a minimax basis marking.
Next, we compute the minimal and maximal explanation vectors of the explicit transition $t_2$ at $M_{b0}$ respectively and derive the other potential minimax basis markings. For instance, for $t_2$, the only minimal explanation vector $\textbf{y}_{\rm min} = \textbf{y}_1 = [1\ 0]^T$ while the only maximal explanation vector $\textbf{y}_{\rm max} = \textbf{y}_2 = [2\ 1]^T$. Since $M_{b0}+C_I\cdot \textbf{y}_{\rm min} + C (\cdot, t_2) = [2\ 0\ 1]^T = M_{b0}$, no new minimax basis marking is generated. However, the pair $(t_2, \textbf{y}_1)$ is stored in ${\rm Tr}_\mathcal{M}$ while the transition relation $(M_{b0}, (t_2, \textbf{y}_1), M_{b0})$ is stored in $\Delta_\mathcal{M}$.
On the other hand, since $M_{b1} = M_{b0}+C_I\cdot \textbf{y}_{\rm max} + C (\cdot, t_2) = [1\ 0\ 0]^T \neq M_{b0}$, let $M_{b1}$ be another minimax basis marking. Corresponding pair and transition relation are also collected.
Analogously, $\textbf{y}_3 = [1\ 0]^T$ and $\mathcal{B_M}$ can be constructed which is graphically shown in Fig. \ref{MBRG}.$\hfill\square$

\begin{figure}[t]
\begin{center}
\includegraphics[width=2.8cm]{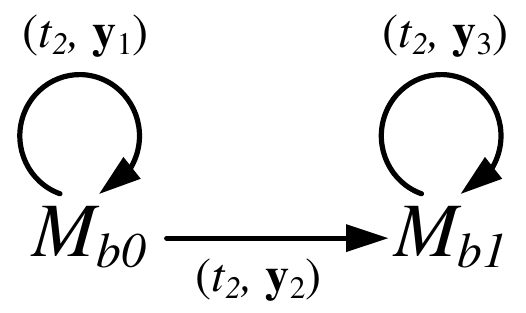}
\caption{The minimax-BRG $\mathcal{B_M}$ with $T_E=\{t_2\}$.}\label{MBRG}
\end{center}
 \end{figure}
\end{example}

%For the sake of space, the expanded BRG $\mathcal{B_E}$ of $G$ is not depicted, which contains 10 expanded basis markings.
%Next, we demonstrate the properties of a minimax-BRG and prove the validity of applying it on verification nonblockingness of a deadlock-free system.
%\subsection{Properties of minimax-BRG}
In the following, we show that the minimax-BRG preserves the reachability information and other non-minimax-basis markings can be algebraically characterized by linear equations.

\begin{proposition}\label{SRGnew}
{\rm Given a marked net $\langle N, M_0\rangle$ with a basis partition $\pi=(T_E, T_I)$ and a marking $M\in\mathbb{N}^m$, $M\in R(N, M_0)$ if and only if
there exists a minimax basis marking $M_b\in \mathcal{M_{B_M}}$ such that $M\in R_I(M_b)$, where $\mathcal{M_{B_M}}$ is the set of the minimax basis markings in minimax-BRG of $\langle N, M_0\rangle$.}
\end{proposition}

\begin{proof}
(only if) It is shown in \cite{c8} that such a property holds for the set of basis markings $\mathcal{M_B}$. Since $\mathcal{M_{B_M}}\supseteq \mathcal{M_B}$, the result follows.

(if) Since $M\in R_I(M_b)$, according to Definition \ref{DefIR}, there exists a firing sequence $\sigma\in T_I^{*}$ such that $M_b[\sigma\rangle M$. On the other hand, there exists another firing sequence $\sigma^{\prime}\in T^*$ such that $M_0[\sigma^{\prime}\rangle M_b$, which implies that $M_0[\sigma^{\prime}\sigma\rangle M$ and concludes the proof.
\end{proof}

In summary, a marking $M$ is reachable from $M_0$ if and only if it belongs to the implicit reach of a minimax basis marking $M_b$ and thus $M$ can be characterized by a linear equation, i.e., $M=M_b+C_I\cdot \textbf{y}_I$, where $\textbf{y}_I=\varphi(\sigma_I)$, $\sigma_I\in T_I^{*}$ and $M_b[\sigma_I\rangle M$.
%
%\begin{remark}\label{Remark1}
%Given a marked net $\langle N, M_0\rangle$ with $\pi=(T_E, T_I)$, $\mathcal{M_{B_M}}$ and $\mathcal{M_{B}}$ are the sets of minimax-basis markings and basis markings of minimax-BRG and BRG of $\langle N, M_0\rangle$, respectively.
%$|\mathcal{M_{B_M}}|=|\mathcal{M_{B}}|=|R(N, M_0)|$ if $T_E=T$ and $T_I=\emptyset$.
%\end{remark}

\section{Verifying Nonblockingness of Bounded Plants Using Minimax-BRGs}\label{NewSection}

In this section, we investigate how minimax-BRGs can be applied to solving the NB-V problem.

\subsection{Unobstructiveness of Minimax-BRGs}\label{sec2.1}
This subsection generalizes the notion of unobstructiveness that is given in \cite{Gu} for a BRG to a minimax-BRG. Such a property is essential to establish our method since it is strongly related to the nonblockingness of a Petri net.
First, we define the set of \textit{i-coreachable minimax basis markings}, denoted by $\mathcal{M_{\rm i_{co}}}$, from which at least one of the final markings in $\mathcal{M_F}$ is reachable by firing implicit transitions only.

%{\blue We first define the set of i-coreachable markings in the set of minimax-basis markings. As a subset of the minimax-basis markings, from which at least one of the final markings in $\mathcal{M_F}$ is reachable by firing implicit transitions only. Markings in $\mathcal{M_{\rm i_{co}}}$ are called \textit{i-coreachable basis markings}.}
%To avoid confusion, instead of using the notation $\mathcal{M_{\rm i_{co}}}$ for BRG, we adopt $\mathcal{M^{\rm \dagger}_{\rm i_{co}}}$ here as a notation.
\begin{definition}\label{DefRe}
{\rm Consider a bounded plant $G=(N, M_0,\\ \mathcal{M_F})$ with the set of minimax basis markings $\mathcal{M_{B_M}}$ in its minimax-BRG.
The set of i-coreachable minimax basis markings of $\mathcal{M_{B_M}}$ is defined as $\mathcal{M_{\rm i_{co}}}=\{M_b\in \mathcal{M_{B_M}}| R_I(M_b)\cap \mathcal{M_F}\neq \emptyset\}.$$\hfill\square$}
\end{definition}

\begin{proposition}\label{pro1027}
{\rm Given a set of final markings defined by a single GMEC $\mathcal{L}_{(\textbf{w},k)}$ and a minimax basis marking $M_b$, $M_b$ belongs to $\mathcal{M_{\rm i_{co}}}$ if and only if the following set of integer constraints is feasible.
%
%GMEC, i.e., $\mathcal{M_F}=\mathcal{L}_{(\textbf{w},k)}=\{M\in\mathbb{N}^m| \textbf{w}^T\cdot M \leq k\}.$

\begin{equation}\label{Equre0.16}
\left\{
             \begin{array}{lr}
            M_b+C_I\cdot \textbf{y}_I=M;\\
            \textbf{w}^T\cdot M \leq k;\\
            \textbf{y}_I\in \mathbb{N}^{n_I};\\
            M\in \mathbb{N}^{m}.
                &
             \end{array}
\right.
\end{equation}}
\end{proposition}

\begin{proof}
(only if) Since $M_b\in \mathcal{M_{\rm i_{co}}}$, according to Definition \ref{DefRe}, $R_I(M_b)\cap \mathcal{M_F}\neq \emptyset$.
Therefore, integer constraints (\ref{Equre0.16}) meets feasible solution $\textbf{y}_I$.

(if) The state equation $M_b+C_I\cdot \textbf{y}_I=M$ provides necessary and sufficient conditions for reachability since the implicit sub-net is acyclic (see Proposition \ref{ProX}). Moreover, $M\in\mathcal{L}_{(\textbf{w},k)}$ is a final marking. Therefore, the statement holds.
\end{proof}

The notion of unobstructiveness in a minimax-BRG is given in Definition \ref{ReDef2}.
In the following, we show how the unobstructiveness of a minimax-BRG is related to the nonblockingness of the corresponding Petri net.

\begin{definition}\label{ReDef2}
{\rm Given a minimax-BRG $\mathcal{B_M}=(\mathcal{M_{B_M}},\\
{\rm Tr_{\mathcal{M}}}, \Delta_{\mathcal{M}}, M_0)$ and a set of i-coreachable minimax basis markings $\mathcal{M_{\rm i_{co}}}\subseteq\mathcal{M_{B_M}}$,
$\mathcal{B_M}$ is said to be \textit{unobstructed} if for all $M_b\in \mathcal{M_{B_M}}$ there exist a marking $M_b^{\prime}\in\mathcal{M_{\rm i_{co}}}$ in $\mathcal{B_M}$ and a firing sequence $\sigma^{+}\in{\rm Tr}_\mathcal{M}^*$ such that $(M_b, \sigma^{+}, {M_b}^{\prime})\in\Delta_\mathcal{M}$.
Otherwise it is \textit{obstructed}.$\hfill\square$}
\end{definition}

\begin{proposition}\label{Pro001}
{\rm Given a plant $(N, M_0, \mathcal{M_F})$, its minimax-BRG is unobstructed if and only if all minimax basis markings are nonblocking.}
%i.e., for all $M_b\in \mathcal{M_B}, R(N, M_b)\cap \mathcal{M_F}\neq\emptyset$.
\end{proposition}

\begin{proof}
(only if) If a minimax-BRG $\mathcal{B_M}$ is unobstructed, then for all $M_b\in \mathcal{M_{B_M}}$ there exist a marking $M_b^{\prime}\in\mathcal{M_{\rm i_{co}}}$ in $\mathcal{B_M}$ and a sequence of pairs $\sigma^{+}=(t_1, \textbf{y}_1), (t_2, \textbf{y}_2), \cdots, (t_k, \textbf{y}_k)\ (\sigma^{+}\in{\rm Tr}_\mathcal{M}^*)$ such that $(M_b, \sigma^{+}, {M_b}^{\prime})\in\Delta_\mathcal{M}$.
By Definition \ref{Def1}, this means that the net admits an evolution: $M_b[\sigma_1 t_1 \sigma_2 t_2 \cdots \sigma_k t_k\rangle {M_b}^{\prime}$, where $\sigma_i\in\varphi^{-1}(\textbf{y}_i)\ (i\in\{1, 2, \cdots, k\})$.
%***********************************************
Since $M_b^{\prime}\in\mathcal{M_{\rm i_{co}}}$, there exists an implicit firing sequence $\sigma_I$ such that $M_b^{\prime}[\sigma_I\rangle M_f$, where $M_f\in \mathcal{M_F}$.
Thus it holds that $M_b[\sigma_1 t_1 \sigma_2 t_2 \cdots \sigma_k t_k\rangle {M_b}^{\prime}[\sigma_I\rangle M_f$, implying that $M_b$ is nonblocking.
(if) We prove this part by contradiction. Assume $\mathcal{B_M}$ is obstructed. Then, there exists $M_b\in \mathcal{M_{B_M}}$ such that $M_b$ is not accessible to any of the i-coreachable marking in $\mathcal{M_{\rm i_{co}}}$ through $\sigma^{+}\in {\rm Tr}^*_\mathcal{M}$.
However, since $M_b$ is nonblocking, there exists a firing sequence $\sigma\in T^*$ and a final marking $M_f\in \mathcal{M_F}$ such that $M_b[\sigma\rangle M_f$.
We write $\sigma=\sigma_1t_{i_1}\cdots\sigma_kt_{i_k}\sigma_{k+1}$ where all $\sigma_i\in T_I^*, t_{i_j}\in T_E, j=1,\ldots, k$.
Following the procedure in the proof of Theorem 3.8 in \cite{c8}, we can repeatedly move transitions in each $\sigma_j$ ($j\in\{1, \ldots, k\}$) to somewhere after $t_{i_j}$ to obtain a new sequence $\sigma_{\rm min,1}t_{i_1}\sigma_{\rm min,2}t_{i_2}\cdots \sigma_{\rm min, k}t_{i_k}\sigma_{k+1}^{\prime}$ such that
\begin{center}
$M_b[\sigma_{\rm min,1}t_{i_1}\rangle M_{b,1}[\sigma_{\rm min,2}t_{i_2}\rangle\cdots[\sigma_{\rm min, k}t_{i_k}\rangle M_{b, k}[\sigma_{k+1}^{\prime}\rangle M_f,$
\end{center}
where each $\sigma_{\rm min,j}\in T_I^*$ is a minimal explanation of $t_{i_j}$ at $M_{b,j}\in \mathcal{M_{B_M}}$ for $j=1,\ldots k$.
Thus, $M_{b, k}\in \mathcal{M_{\rm i_{co}}}$, a contradiction.
\end{proof}

%{\red According to Proposition \ref{Pro001}, to determine the unobstructiveness of minimax-BRG $\mathcal{B_M}$, it is only required to check if all minimax basis markings in $\mathcal{M_{B_M}}$ are nonblocking. To verify it, one needs to analyze the minimax BRG by checking if all minimax basis markings are co-reachable to some i-coreachable minimax basis markings. An example is illustrated in the following.}
According to Proposition \ref{Pro001}, to determine the unobstructiveness of minimax-BRG $\mathcal{B_M}$, it is only required to check if all minimax basis markings are co-reachable to some i-coreachable minimax basis markings in $\mathcal{B_M}$.
This can be done by using some search algorithm (e.g., \textit{Dijkstra}) in the underlying digraph of the minimax-BRG, whose complexity is polynomial in $\mathcal{B_M}$.
%{\blue This can be technically coped by regarding $\mathcal{B_M}$ as a \textit{digraph} $\mathcal{G}$ and the set of i-coreachable markings as the set of targeted nodes in $\mathcal{G}$.
%Then, one may need to check for all non-targeted nodes that if there exists a directed path from a non-target node to at least a targeted one. Implemented by the searching approaches (e.g., depth-first search or breadth-first search), the complexity of the above-mentioned task is polynomial in the size of $\mathcal{G}$.}
An example is illustrated in the following to help clarify Proposition \ref{Pro001}.

%%%%Original%%%%
%According to Proposition \ref{Pro001}, to determine the unobstructiveness of minimax-BRG $\mathcal{B_M}$, we need to check if all minimax basis markings in $\mathcal{M_{B_M}}$ are nonblocking only, which can be verified by checking if all minimax basis markings are co-reachable to some i-coreachable minimax basis markings by analyzing the minimax BRG. An example is illustrated in the following.
%%%%Original%%%%

%%%%%%%%%%%%%%%%%%%%%%Original Example 3%%%%%%%%%%%%%%%%%%%%%%
%\begin{example}\label{Example3}
%{\rm Consider again the marked net $\langle N, M_0\rangle$ shown in Fig. \ref{SO1} and discussed in Example \ref{EXPSRG} with $M_0=[3\ 0\ 1\ 1]^{\rm T}$ and $T_E=\{t_2, t_4\}$.
%Assuming that the set of final markings is $\mathcal{M_F}=\mathcal{L}_{(\textbf{w},k)}$ where $\textbf{w}=[1\ 1\ 0\ 0]^{\rm T}$ and $k=1$, we want to verify the unobstructiveness of its minimax-BRG $\mathcal{B_M}$ shown in Fig. \ref{MBRG}.
%
%First we need to determine the set of i-coreachable minimax basis markings of this system by solving ILPP (\ref{Equre0.16}): we conclude that $\mathcal{M_{\rm i_{co}}}=\{[0\ 1\ 0\ 0]^{\rm T}, [1\ 0\ 0\ 0]^{\rm T}\}$.
%Since all minimax basis markings are co-reachable to a marking in $\mathcal{M_{\rm i_{co}}}$ in $\mathcal{B_M}$, the minimax-BRG is unobstructed}.$\hfill\square$
%\end{example}
%%%%%%%%%%%%%%%%%%%%%%Original Example 3%%%%%%%%%%%%%%%%%%%%%%

\begin{example}\label{Example3}
{\rm Consider again the parameterized plant $(N, M_0, \mathcal{M_F})$ in Fig. \ref{Fig1} (left) with $\alpha=1$, $T_E=\{t_2\}$ and $\mathcal{M_F}=\mathcal{L}_{(\textbf{w},k)}$ where $\textbf{w}=[1\ 1\ 0\ 0]^{\rm T}$ and $k=1$.
We explain how to verify the unobstructiveness of its minimax-BRG $\mathcal{B_M}$ shown in Fig. \ref{MBRG}.
By solving the linear constraint (\ref{Equre0.16}) in Proposition \ref{pro1027}, we conclude that $\mathcal{M_{\rm i_{co}}}=\{[2\ 0\ 1]^{\rm T}\}$.
Since there is no directed path from $M_{b1}$ to $M_{b0}$, $M_{b1}$ is not co-reachable to the only marking in $\mathcal{M_{\rm i_{co}}}$. Thus, the minimax-BRG $\mathcal{B_M}$ in Fig. \ref{MBRG} is obstructed}.$\hfill\square$
\end{example}

\subsection{Verifying Nonblockingness of Deadlock-Free Plants}\label{DeadlockFree}
%\subsection{\red Verifying Nonblockingness of Deadlock-Free PNSs Using Minimax-BRGs}
In this subsection, we focus on deadlock-free plants. An intermediate result is shown in Proposition \ref{CoreTheorem}.

\begin{proposition}\label{CoreTheorem}
{\rm Given a bounded marked net $\langle N, M_0\rangle$ with basis partition $\pi=(T_E, T_I)$, for all $M\in R(N, M_0)$, for all $t\in T_E$, for all $\sigma\in \Sigma(M, t)$ with $M[\sigma t\rangle M^{\prime}$, the following implication holds:
\begin{equation}
\begin{aligned}
(\forall \sigma^{\prime}\in\Sigma(M, t))\ \varphi(\sigma)-\varphi(\sigma^{\prime})&=\tilde{\textbf{y}}\geq \textbf{0} \Rightarrow\\
(\exists \sigma^{\prime\prime}\in&\varphi^{-1}(\tilde{\textbf{y}}))\ M[\sigma^{\prime}t \sigma^{\prime\prime}\rangle M^{\prime}
\\
\end{aligned}
\end{equation}}
\end{proposition}
%\begin{equation}
%\begin{aligned}
%(\forall\sigma\in \Sigma(M, t))\land (M[\sigma t\rangle \bar{M})\Rightarrow&\\
%(\forall \sigma^{\prime}\in\Sigma(M, t))\ & \varphi(\sigma)-\varphi(\sigma^{\prime})=\tilde{\textbf{y}}\geq \textbf{0} \\
%(\exists \sigma^{\prime\prime}\in\varphi^{-1}(\tilde{\textbf{y}}))&\ M[\sigma^{\prime}t \sigma^{\prime\prime}\rangle \bar{M}
%\\
%\end{aligned}
%\end{equation}
%\end{proposition}

\begin{proof}
Let $M^{\prime\prime}\in \mathbb{N}^{m}$ such that $M[\sigma^{\prime} t\rangle M^{\prime\prime}$. Then it holds that:
\begin{equation}\label{EQUNEW1}
\left\{
             \begin{array}{lr}
            M^{\prime}=M+C_I\cdot \varphi(\sigma)+C(\cdot, t) \\
            M^{\prime\prime}=M+C_I\cdot \varphi(\sigma^{\prime})+C(\cdot, t)
             \end{array}
\right.
\end{equation}

From Equation (\ref{EQUNEW1}) we conclude that $M^{\prime}-M^{\prime\prime}=C_I(\varphi(\sigma)-\varphi(\sigma^{\prime}))$, which implies $M^{\prime}=M^{\prime\prime}+C_I\cdot \tilde{\textbf{y}}$ and $\tilde{\textbf{y}}\in \mathbb{N}^{n}$. This indicates:
\begin{equation}
\begin{aligned}
\exists\sigma^{\prime\prime}\in \varphi^{-1}(\tilde{\textbf{y}}): M^{\prime\prime}[\sigma^{\prime\prime}\rangle M^{\prime}
\end{aligned}
\end{equation}
and thus $M[\sigma^{\prime}t\rangle M^{\prime\prime}[\sigma^{\prime\prime}\rangle M^{\prime}$ that concludes the proof.
%Let $M^{\prime\prime}\in \mathbb{N}^{m}$ such that $M[\sigma^{\prime} t\rangle M^{\prime\prime}$. Then it holds that:
%\begin{equation}\label{EQUNEW1}
%\left\{
%             \begin{array}{lr}
%            M^{\prime}=M+C_I\cdot \varphi(\sigma)+C(\cdot, t) \\
%            M^{\prime\prime}=M+C_I\cdot \varphi(\sigma^{\prime})+C(\cdot, t)
%             \end{array}
%\right.
%\end{equation}
%
%From Equation (\ref{EQUNEW1}) we conclude that $M^{\prime}-M^{\prime\prime}=C_I(\varphi(\sigma)-\varphi(\sigma^{\prime}))$, which implies $M^{\prime}=M^{\prime\prime}+C_I\cdot \tilde{\textbf{y}}$ and $\tilde{\textbf{y}}\in \mathbb{N}^{n}$. This indicates:
%\begin{equation}
%\begin{aligned}
%\exists\sigma^{\prime\prime}\in \varphi^{-1}(\tilde{\textbf{y}}): M^{\prime\prime}[\sigma^{\prime\prime}\rangle M^{\prime}
%\end{aligned}
%\end{equation}
%and thus $M[\sigma^{\prime}t\rangle M^{\prime\prime}[\sigma^{\prime\prime}\rangle M^{\prime}$ that concludes the proof.
\end{proof}

Proposition \ref{CoreTheorem} shows the connection between two markings $M^{\prime}$ and $M^{\prime\prime}$ reachable from $M\in R(N, M_0)$, i.e., $M[\sigma t\rangle M^{\prime}$ and $M[\sigma^{\prime} t\rangle M^{\prime\prime}$, where $t\in T_E$, $\sigma\in\Sigma(M, t)$, $\sigma^{\prime}\in\Sigma(M, t)$, and $\varphi(\sigma)-\varphi(\sigma^{\prime})=\tilde{\textbf{y}}\geq 0$.
If $M^{\prime}$ is nonblocking, then $M^{\prime\prime}$ is nonblocking as well, since there exists a firing sequence $\sigma^{\prime\prime}\in \varphi^{-1}(\tilde{\textbf{y}})$ such that $M^{\prime\prime}[\sigma^{\prime\prime}\rangle M^{\prime}$.
According to this proposition, we next show that the unobstructiveness of the minimax-BRG is a necessary and sufficient condition for nonblockingness of a net in the considered class.

\begin{lemma}\label{NewLemma}
{\rm Consider a bounded deadlock-free marked net $\langle N, M_0\rangle$ with a basis partition $\pi=(T_E, T_I)$.
For all markings $M\in R(N, M_0)$, there exists a firing sequence $\sigma t$, where $\sigma \in T_I^{*}$ and $t\in T_E$, such that $M[\sigma t\rangle$ holds.}
\end{lemma}

\begin{proof}
We prove this statement by contradiction.
Assume the system is deadlock-free and there exists a marking $M$ from which all explicit transitions are not enabled.
Since the implicit sub-net of the system is bounded and acyclic, the maximal length of sequences enabled at $M$ and composed by only implicit transitions is finite.
Hence, from $M$, after the firing of such maximal sequences of implicit transitions, the net reaches a deadlock, which is a contradiction.
\end{proof}

The result in Lemma \ref{NewLemma} can be applied to both BRG and minimax-BRG.
However, it does not imply that the marking reached after the firing of the explicit transition is a basis marking, as we have shown in Example \ref{E1}. Hence, it does not rule out the presence of livelocks in the BRG.

\begin{lemma}\label{NewLemma2}
{\rm Consider a bounded deadlock-free marked net $\langle N, M_0\rangle$ with a basis partition $\pi=(T_E, T_I)$.
For all markings $M\in R(N, M_0)$, for all explicit transition $t\in T_E$, the following holds:
\begin{center}
$\sigma\in\Sigma(M,t)\Rightarrow(\exists\sigma^{\prime}\in \Sigma_{\rm max}(M,t))\ \varphi(\sigma^{\prime})\geq\varphi(\sigma).$
\end{center}}
\end{lemma}

\begin{proof}
If $\sigma\notin\Sigma_{\rm max}(M,t)$, according to Definition \ref{DEFMAX}, there exists an explanation $\sigma^{\prime}\in \Sigma_{\rm max}(M,t))$ such that $\varphi(\sigma^{\prime})>\varphi(\sigma)$; otherwise $\varphi(\sigma^{\prime})=\varphi(\sigma)$, hence the result holds.
\end{proof}

\begin{lemma}\label{NewLemma3}
{\rm Given a bounded deadlock-free marked net $\langle N, M_0\rangle$ with a basis partition $\pi=(T_E, T_I)$, the set of minimax basis markings of the system is $\mathcal{M_{B_M}}$.
For all markings $M\in R(N, M_0)$, there exists $M_b\in \mathcal{M_{B_M}}$ such that $M_b\in R(N, M)$.}
\end{lemma}

\begin{proof}
Due to Lemma \ref{NewLemma}, there exists a firing sequence $\sigma t$, where $\sigma \in T_I^{*}$ and $t\in T_E$, such that $M[\sigma t\rangle$, which implies that $\sigma\in \Sigma(M, t)$.
By Lemma \ref{NewLemma2}, there exists a maximal explanation $\sigma^{\prime}\in \Sigma_{\rm max}(M,t)$ such that $\varphi(\sigma^{\prime})\geq\varphi(\sigma)$.
Let $\varphi(\sigma^{\prime})-\varphi(\sigma)=\textbf{y}$ and $M[\sigma^{\prime} t\rangle M^{\prime}$, by Definition \ref{Def1}, $M^{\prime}\in\mathcal{M_{B_M}}$.
According to Proposition \ref{CoreTheorem}, there exists a firing sequence $\sigma^{\prime\prime}\in\varphi^{-1}(\textbf{y})$ such that $M[\sigma t\sigma^{\prime\prime}\rangle M^{\prime}$, which implies that $M^{\prime}\in R(N,M)$.
\end{proof}

\begin{theorem}\label{Theoremminimax}
{\rm A bounded deadlock-free plant $G=(N, M_0, \mathcal{M_F})$ is nonblocking if and only if its minimax-BRG $\mathcal{B_M}$ is unobstructed.}
\end{theorem}

\begin{proof}
(only if) Since the net is nonblocking, all reachable markings, including all minimax basis markings, are nonblocking. By Proposition \ref{Pro001}, its minimax-BRG $\mathcal{B_M}$ is unobstructed.

(if) Consider an arbitrary marking $M\in R(N, M_0)$. By Lemma \ref{NewLemma3}, there exists a minimax basis marking $M_b\in \mathcal{M_{B_M}}$ such that $M_b\in R(N, M)$, i.e., there exists a firing sequence $\sigma\in T^*$ such that $M[\sigma\rangle M_b$.
Since the minimax BRG $\mathcal{B_M}$ is unobstructed, according to Proposition \ref{Pro001}, all minimax basis markings including $M_b$ are nonblocking, which implies that marking $M$ is co-reachable to a nonblocking marking.
Hence, $G$ is nonblocking.
\end{proof}

By Theorem \ref{Theoremminimax}, for a deadlock-free net, one can use an arbitrary basis partition to construct the minimax-BRG to verify its nonblockingness. Since the existence of a livelock component that contains all blocking markings implies the existence of at least a blocking minimax basis marking $M_b$ in $\mathcal{B_M}$, the potential livelock problem mentioned in Section \ref{sec2} is avoided.

\subsection{Verifying Nonblockingness of Plants with Deadlocks}\label{secVD}

In this subsection, we generalize the results in Section \ref{DeadlockFree} to systems that are not deadlock-free.
Notice that a dead marking $M\in R(N, M_0)$ can either be non-final (i.e., $M\notin \mathcal{M_F}$) or final (i.e., $M\in\mathcal{M_F}$).

\begin{theorem}\label{0904Th}
{\rm A bounded plant $G=(N, M_0, \mathcal{M_F})$ is nonblocking if and only if its minimax-BRG $\mathcal{B_{M}}$ is unobstructed and all its dead markings are final.}
\end{theorem}

\begin{proof}
(only if) When all reachable markings are nonblocking, all dead markings (if any exists) and all minimax basis markings are also nonblocking. Hence, all dead markings are final and by Proposition \ref{Pro001}, the minimax-BRG $\mathcal{B_M}$ is unobstructed.

(if)
If the minimax-BRG $\mathcal{B_{M}}$ is unobstructed, all minimax basis markings are nonblocking, by Proposition \ref{Pro001}.
%For a marking $M\in R(N, M_0)$, according to Lemma \ref{NewLemma2}, there exists a firing sequence $\sigma t$ where $\sigma\in T_I^{*}$ and $t\in T_E$, such that $M[\sigma t\rangle M^{\prime}$ or there exists a firing sequence $\sigma^{\prime}\in T_I^{*}$, such that $M[\sigma^{\prime}\rangle M^{\prime\prime}$ and there does not exist a transition $t^{\prime}\in T_E$ such that $M^{\prime\prime}[t^{\prime}\rangle$.
Consider an arbitrary marking $M\in R(N, M_0)$. By Proposition \ref{SRGnew}, there exist a minimax basis marking $M_b\in\mathcal{M_{B_M}}$ in the minimax-BRG of the system and an implicit firing sequence $\sigma_I\in T_I^{*}$ such that $M_b[\sigma_I\rangle M$.

We prove that marking $M$ is nonblocking by contradiction. In fact, if we assume that $M$ is blocking, since all dead markings are final, $M$ is neither dead nor co-reachable to a deadlock in the system.
%Hence, there exists a firing sequence $\sigma\in T^*$ such that $M[\sigma\rangle$.
Suppose that from $M$ no explicit transition can eventually fire: following the argument of the proof of Lemma \ref{NewLemma}, a dead marking will be reached, leading to a contradiction. Therefore, there exist $\sigma_I^{\prime}\in T_I^*$ and $t\in T_E$ such that $M[\sigma_I^{\prime}t\rangle$ and thus $M_b[\sigma_I\sigma_I^{\prime}t\rangle$. Also, there exists a maximal explanation $\sigma^{\prime}\in\Sigma_{\rm max}(M_b, t)$ such that $\varphi(\sigma^{\prime})\geq\varphi(\sigma_I\sigma_I^{\prime})$. According to Proposition \ref{CoreTheorem}, it follows that $M$ is co-reachable to a minimax basis marking, which implies that $M$ is nonblocking, another contradiction, which concludes the proof.
\end{proof}
According to Theorem \ref{0904Th}, determining the nonblockingness of a plant $G$ can be addressed by two steps:
%\begin{enumerate}
%  \item determine if there exists a reachable non-final dead marking. If yes, then $G$ is blocking; otherwise go to step 2;
%  \item determine if a minimax BRG of $G$ is unobstructed. If yes,
%\end{enumerate}
(1) determine if there exists a reachable non-final dead marking; if not, then (2) determine the unobstructiveness of a minimax-BRG of it.
%
%************************************************************************************************************************
%************************************************************************************************************************
%{\red Note that in step (1), a non-final dead marking can be either a minimax basis marking or a non-minimax-basis one.
%Since a minimax basis marking can be verified in step (2), for step (1), we aim to compute the set of \textit{non-minimax-basis dead markings} $\mathcal{D}$ (defined in Definition \ref{Defnmbd}) as a preliminary step. If set $\mathcal{D}$ including non-final markings, the system is blocking; otherwise step (2) should be further executed.}
%************************************************************************************************************************
%************************************************************************************************************************
%{\red Thanks to the implication of deadlock characterization in acyclic nets presented in Section \ref{NewSub0208}, we introduce how to determine the existence of non-final dead markings by using minimax-BRG in Section \ref{Sub5.1}.}

Since step (2) has already been discussed in the previous section, we only need to study step (1).
Next, we show how to determine the existence of non-final dead markings by using the minimax-BRG.
Denote the set of non-final dead markings as $\mathcal{D}_{\rm nf}$.
Then, we define the set of \textit{maximal implicit firing sequences} and the corresponding set of vectors as follows.

\begin{definition}\label{1220new}
{\rm Given a bounded marked net $\langle N, M_0\rangle$ with basis partition $\pi=(T_E, T_I)$ and a marking $M\in R(N, M_0)$, we define
\begin{center}
$\Sigma_{\rm I, max}(M)=\{\sigma\in T_I^*| (M[\sigma\rangle)\wedge(\nexists \sigma^{\prime}\in T_I^*: M[\sigma^{\prime}\rangle, \varphi(\sigma^{\prime})\gneqq \varphi(\sigma))\}$
%\forall t\in T_E, \sigma\notin \Sigma_{{\rm max}}(M, t)\}$
\end{center}
as the set of maximal implicit firing sequences at $M$, and
\begin{center}
$Y_{\rm I, max}(M)=\{\varphi(\sigma)\in \mathbb{N}^{n_{I}}| \sigma\in \Sigma_{\rm I, max}(M)\}$
\end{center}
as the corresponding set of maximal implicit firing vectors.$\hfill\square$}
\end{definition}

\begin{proposition}\label{Prop0121}
{\rm Given a bounded marked net $\langle N, M_0\rangle$ with basis partition $\pi=(T_E, T_I)$, let $\mathcal{M_{B_M}}$ be its minimax basis marking set.
Marking $M\in R(N, M_0)$ is dead if and only if there exist $M_b\in\mathcal{M_{B_M}}$ and $\sigma\in\Sigma_{\rm I, max}(M_b)$ such that for all $t\in T_E, \sigma\notin \Sigma_{{\rm max}}(M_b, t)$ and $M_b[\sigma\rangle M.$}
%there exists a minimax-basis marking $M_b\in\mathcal{M_{B_M}}$ and a maximal implicit firing sequence $\sigma\in \Sigma_{\rm I, max}(M_b)$ such that $M_b[\sigma\rangle M$ holds.}
\end{proposition}
%\begin{corollary}\label{Prop0121}
%{\rm Given a bounded marked net $\langle N, M_0\rangle$ with basis partition $\pi=(T_E, T_I)$, let $\mathcal{M_{B_M}}$ be its minimax basis marking set.
%{\red Marking $M\in R(N, M_0)\setminus\mathcal{M_{B_M}}$ is dead if and only if there exist $M_b\in\mathcal{M_{B_M}}$ and $\sigma\in\Sigma_{\rm I, max}(M_b)$ such that for all $t\in T_E, \sigma\notin \Sigma_{{\rm max}}(M_b, t)$ and $M_b[\sigma\rangle M.$}}
%%there exists a minimax-basis marking $M_b\in\mathcal{M_{B_M}}$ and a maximal implicit firing sequence $\sigma\in \Sigma_{\rm I, max}(M_b)$ such that $M_b[\sigma\rangle M$ holds.}
%\end{corollary}

\begin{proof}
(if) Since $\sigma\in \Sigma_{\rm I, max}(M_b)$, there does not exist an implicit transition $t_I\in T_I$ such that $M[t_I\rangle$. On the other hand, since for all $t\in T_E, \sigma\notin \Sigma_{{\rm max}}(M_b, t)$, i.e., there does not exist an explicit transition $t^{\prime}\in T_E$ such that $M[t^{\prime}\rangle$, which implies that $M$ is dead.

(only if) Since $M$ is dead, there does not exist $t\in T$ such that $M[t\rangle$.
Therefore, there exists $\sigma\in\Sigma_{\rm I, max}(M_b)$ such that for all $t\in T_E, \sigma\notin \Sigma_{{\rm max}}(M_b, t)$ and $M_b[\sigma\rangle M$.
\end{proof}
%\begin{proof}
%%This statement follows from Propositions \ref{SRGnew} and \ref{Lemma0207}.
%{\red Trivially holds based on Propositions \ref{SRGnew} and \ref{Lemma0207}.}
%\end{proof}
Proposition \ref{Prop0121} shows the relation between dead markings and minimax basis markings in a bounded system, i.e., all reachable dead markings can be obtained by firing a maximal implicit firing sequence $\sigma$ from a minimax basis marking $M_b$ where for all $t\in T_E$, $\sigma$ is not a maximal explanation of $t$. %by first computing all minimax basis markings and then for all $M_b\in\mathcal{M_{B_M}}$, we compute the set of firing maximal implicit sequence vectors $\textbf{y}$ with for all $t\in T_E$, $\textbf{y}\notin Y_{{\rm max}}(M_b, t)$. As a result, $M+C_I\cdot\textbf{y}$ is a non-minimax-basis dead marking.
%*****************************************************************************************************
%However, the set of dead markings that are minimax-basis may be omitted: suppose that $M_b^{\prime}\in\mathcal{M_{B_M}}$ is dead in a bounded system, then $\Sigma_{\rm I, max}(M_b^{\prime})=\emptyset$ and for all $t\in T_E, \Sigma_{{\rm max}}(M_b^{\prime}, t)=\emptyset$.
%*****************************************************************************************************
%Hence, we introduce Algorithm \ref{Algonew} to compute not only non-minimax-basis but also minimax-basis dead markings and determine the if $\mathcal{D}_{\rm nf}=\emptyset$ in a plant.
Next, we introduce Algorithm \ref{Algonew} to verify if there exist non-final dead markings in a plant.

\begin{algorithm}[t]
%*******************************************************************************************
%************************Non final deadlock determination algorithm*************************
%*******************************************************************************************
\caption{Verification of $\mathcal{D}_{\rm nf}$} %
\begin{algorithmic}[1]
\REQUIRE A bounded plant $(N, M_0, \mathcal{M_F})$ with $\pi=(T_E ,T_I)$ and its minimax basis marking set $\mathcal{M_{B_M}}$
\ENSURE ``$\mathcal{D}_{\rm nf}=\emptyset$''$\slash$ ``$\mathcal{D}_{\rm nf}\neq\emptyset$''
\STATE {$\mathcal{D}_{\rm nf}:=\emptyset$, $T^{\prime}:=T\cup\{t_0\}$ and $T_E^{\prime}:=T_E\cup\{t_0\}$};
%\STATE Let $Pre(\cdot, t_0)=Post(\cdot, t_0)=\textbf{0}$, $Pre^{\prime}=[Pre(\cdot, t_0);Pre]$, and $Post^{\prime}=[Post(\cdot, t_0);Post]$;
\STATE {$Pre^{\prime}:=[\textbf{0};Pre]$ and $Post^{\prime}:=[\textbf{0};Post]$};
\STATE {$N^{\prime}:=(P, T^{\prime}, Pre^{\prime}, Post^{\prime})$ and $\pi^{\prime}:=(T_E^{\prime} ,T_I)$};
\STATE Construct a bounded plant $(N^{\prime}, M_0, \mathcal{M_F})$ with basis partition $\pi^{\prime}$;
\FORALL {$M\in\mathcal{M_{B_M}}$,}
\FORALL {$\textbf{y}\in Y_{\rm max}(M, t_0)$,}
\STATE {$M^{\prime}:=M+C_I\cdot \textbf{y}$};
\IF {$M^{\prime}$ is dead and $M^{\prime}\notin \mathcal{M_F}$,}
\STATE {$\mathcal{D}_{\rm nf}:=\mathcal{D}_{\rm nf}\cup\{M^{\prime}\}$};
\STATE Output ``$\mathcal{D}_{\rm nf}\neq\emptyset$'' and Return;
\ENDIF
\ENDFOR
\ENDFOR
\IF {$\mathcal{D}_{\rm nf}=\emptyset$,}
\STATE {Output ``$\mathcal{D}_{\rm nf}=\emptyset$'' and Return.}
\ENDIF
\end{algorithmic}\label{Algonew}
\end{algorithm}

\begin{table*}[t]
\caption{Analysis of the reachability graph, expanded BRG from \cite{Gu} and minimax-BRG for the plant in Fig. \ref{FMS} with $T_E=\{t_3, t_6, t_{11}, t_{13}\}$.}\label{tablenew1}
\scalebox{0.66}{
\begin{threeparttable}
\begin{tabular}{c|c|c||c|c||c|c||c|c|c|c|c|c||c||c||c}
\toprule[1pt]
  Run & $\lambda$ & $\mu$ & $|R(N, M_0)|$ & Time\ (s) & $|\mathcal{M_{B_E}}|$ & Time\ (s) & $|\mathcal{M_{B_M}}|$ & Time\ (s) & $\mathcal{D}_{\rm nf}=\emptyset$? & Time\ (s) & Unobstructed? & Time\ (s) & NB? & $|\mathcal{M_{B_M}}|/|\mathcal{M_{B_E}}|$ & $|\mathcal{M_{B_M}}|/|R(N, M_0)|$\\
  \hline
  1 & 5 & 1 & 102 & $<1$ & 31 & $0.2$ & 11 & 0.04 & Yes & 0.03 & Yes & 1.9 & Yes & $35.5\%$ & $10.8\%$\\
  \hline
  2 & 5 & 2 & 384 & 1 & 191 & 0.7 & 37 & 0.2 & Yes & 0.1 & Yes & 6 & Yes & $19.3\%$ & $9.6\%$\\
  \hline
  3 & 5 & 3 & 688 & 2 & 405 & 1 & 68 & 0.4 & Yes & 0.4 & Yes & 12 & Yes & $16.8\%$ & $9.9\%$\\
  \hline
  4 & 6 & 1 & 840 & 4 & 449 & 2 & 81 & 0.5 & Yes & 0.5 & Yes & 15 & Yes & $18.0\%$ & $9.6\%$\\
  \hline
  5 & 6 & 2 & 12066 & 431 & 9117 & 302 & 1171 & 23 & Yes & 16 & No & 251 & No & $12.8\%$ & $9.7\%$\\
%  \hline
%  6 & 6 & 3 & 88681 & 24354 & 75378 & 20944 & 9985 & 833 & Yes & 249 & - & - & No & $13.2\%$ & $11.3\%$\\
    \hline
  6 & 6 & 3 & 88681 & 24354 & 75378 & 20944 & 9985 & 833 & No & 58 & - & - & No & $13.2\%$ & $11.3\%$\\
  \hline
  7 & 6 & 4 & - & o.t. & - & o.t. & 22095 & 4517 & Yes & 1099 & Yes & 3502 & Yes & - & - \\
  \hline
  8 & 6 & 5 & - & o.t. & - & o.t. & 31147 & 10082 & No & 955 & - & - & No & - & - \\
  \hline
  9 & 6 & 6 & - & o.t. & - & o.t. & 41817 & 18295 & No & 1618 & - & - & No & - & - \\
  \hline
  10 & 6 & 7 & - & o.t. & - & o.t. & 45458 & 21229 & No & 1754 & - & - & No & - & - \\
%Time expiration data    8 & 7 & 5 & - & o.t. & - & o.t. & 39781 & o.t. & - & - & -\\
  \bottomrule[1pt]
\end{tabular}
\begin{tablenotes}
        \item[*] The computing time is denoted by \textit{overtime} (o.t.) if the program does not terminate within 28,800 seconds (8 hours).
      \end{tablenotes}
\end{threeparttable}}
\end{table*}

In Algorithm \ref{Algonew}, first, from lines 1$\--$4, we add an explicit transition $t_0$ to $N$ with $Pre(\cdot, t_0)=Post(\cdot, t_0)=\textbf{0}$ and derive a new plant $(N^{\prime}, M_0, \mathcal{M_F})$. Obviously, $t_0$ is enabled from any reachable marking and, since its firing does not modify the marking, it holds that $R(N, M_0)=R(N^{\prime}, M_0)$.
Hence, for all $M_b\in\mathcal{M_{B_M}}$, we conclude that $Y_{\rm I, max}(M_b)=Y_{\rm max}(M_b, t_0)$, i.e., the set of maximal implicit firing vectors at $M_b$ can be determined by computing maximal explanation of $t_0$ at $M_b$ based on Algorithm \ref{AlgoMax}.

Then, we determine if, for all $t\in T_E$, the obtained firing vector $\textbf{y}\in Y_{\rm I, max}(M_b)$ is not an explanation of $t$ at $M_b$. Implemented in lines 5$\--$16, this consists in checking if, for all $t\in T_E$, $t$ is disabled at marking $M^{\prime}=M_b+C_I\cdot\textbf{y}$: since no implicit transition can fire at $M^{\prime}$, the only transitions that can possibly fire are those explicit ones.
If no explicit transition is enabled at $M^{\prime}$, according to Proposition \ref{Prop0121}, marking $M^{\prime}$ is dead.
Further, $M^{\prime}$ will be added into the set $\mathcal{D}_{\rm nf}$ if it is dead and not final.
Note that Algorithm \ref{Algonew} also tests if a minimax basis marking $M_b^{\prime}\in\mathcal{M_{B_M}}$ is dead. Since $Y_{\rm max}(M_b^{\prime}, t_0)=\{\textbf{0}\}$ and for all $t\in (T_E^{\prime}\setminus\{t_0\})$, $Y_{\rm max}(M_b^{\prime}, t)=\emptyset$, $M_b^{\prime}$ will be added to $\mathcal{D}_{\rm nf}$ if it is not final.
When the algorithm terminates, if $\mathcal{D}_{\rm nf}\neq\emptyset$, we conclude that the plant is blocking; otherwise, the unobstructiveness verification procedure (mentioned in Section \ref{sec2.1}) of the minimax-BRG should be further executed.

The complexity of Algorithm \ref{Algonew} depends on the two \textit{for} loops (lines 5$\--$13).
 First, there are $|\mathcal{M_{B_M}}|$ and $|Y_{\rm max}(M, t_0)|$ iterations in lines 5 and 6, respectively.
  In line 8, to verify $M^{\prime}$ is dead, one may need to test if $M^{\prime}\ngeq Pre^{\prime}(\cdot, t)$ for all $t\in T_E$ (no need to test transitions in $T_I$ since no implicit transition is enabled at $M^{\prime}$), which requires $|T_E|$ iterations.
  In summary, the worst-case time complexity of Algorithm \ref{Algonew} is $\mathcal{O}(|\mathcal{M_{B_M}}|\cdot|Y_{\rm max}(M, t_0)|\cdot|T_E|)$.

\section{Case Studyz}\label{sec2.3}

\begin{figure}[t]
\begin{center}
\includegraphics[width=7cm]{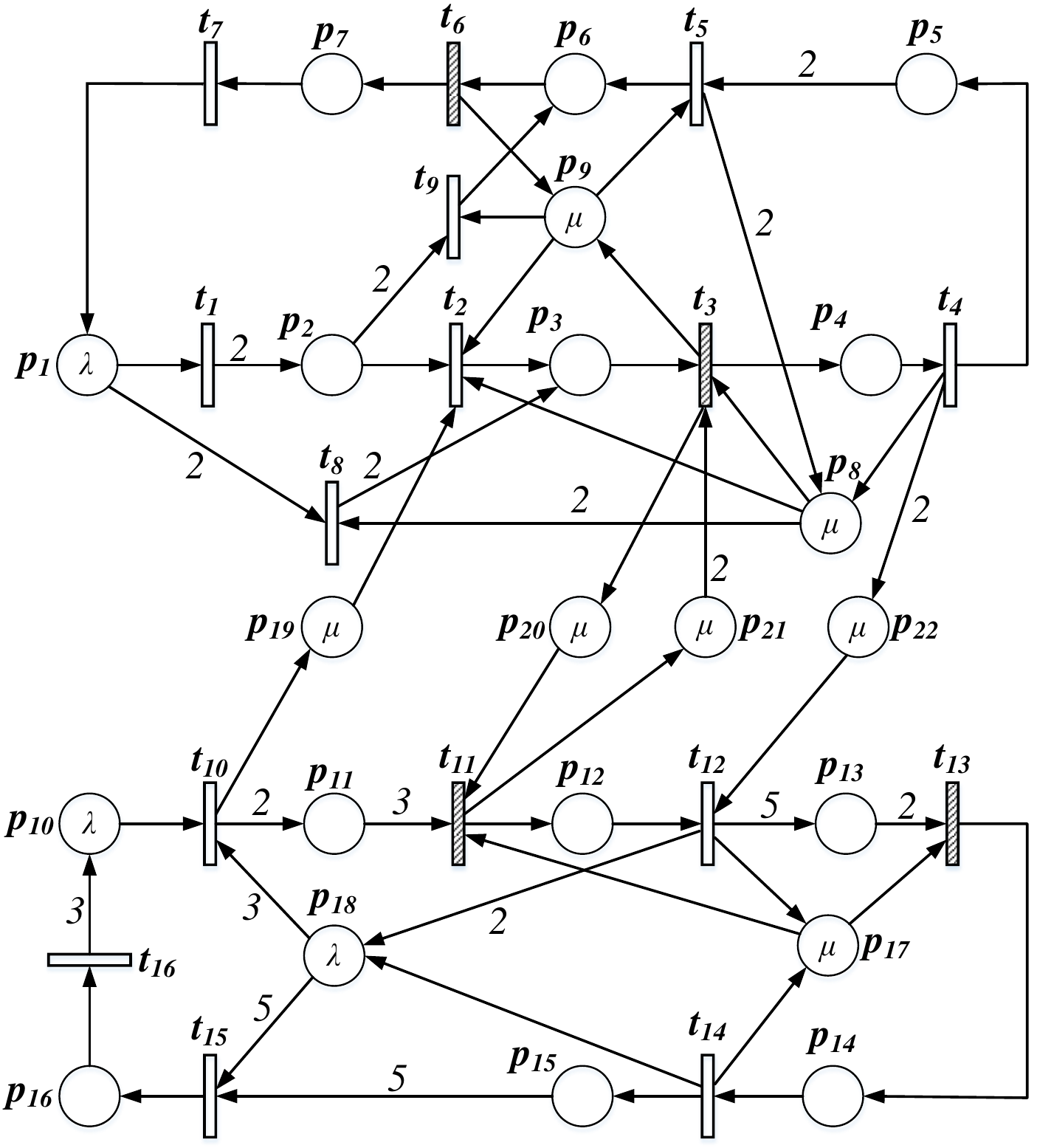}
\caption{A parameterized manufacturing example.}\label{FMS}
\end{center}
\end{figure}

%The second benchmark shows the scalability of our approach.
%\begin{table}[b]
%\caption{Computation of non-final dead markings and total time of nonblockingness determination of the system in Fig. \ref{FMS}.}\label{tablenew2}
%\scalebox{0.65}{
%\begin{threeparttable}
%\begin{tabular}{c|c|c||c|c||c||c}
%\toprule[1pt]
%  Run & $\lambda$ & $\mu$ & $|\mathcal{D}_{\rm nf}|$ & Time\ (s) & Total Time\ (s)  & Blocking?\\
%  \hline
%  1 & 5 & 1 & 1 & 0.04 & 0.08 & Yes\\ %%FUll Time ILPP: 7
%  \hline
%  2 & 5 & 2 & 1 & 0.17 & 0.37 & Yes\\ %%FUll Time ILPP: 16
%  \hline
%  3 & 5 & 3 & 1 & 0.37 & 0.77 & Yes\\   %%FUll Time ILPP: 26
%  \hline
%  4 & 6 & 1 & 1 & 0.5 & 1 & Yes\\  %%FUll Time ILPP: 32
%  \hline
%  5 & 6 & 2 & 2 & 16 & 39 & Yes\\ %%FUll Time ILPP: 449
%  \hline
%  6 & 6 & 3 & 4 & 246 & 1079 & Yes\\ % Time(\mathcal{M_E}): 65.67   %%FUll Time ILPP: 3608
%  \hline
%  7 & 6 & 4 & 6 & 1123 & 5640 & Yes\\  % Time(\mathcal{M_E}): 132   %%FUll Time ILPP: 8304 %% Sort: 20455
%  \hline
%  8 & 6 & 5 & 11 & 2030 & 12112 & Yes\\
%  \hline
%  9 & 6 & 6 & 18 & 3495 & 21790 & Yes\\  % Time(\mathcal{M_E}): 142
%  \hline
%  10 & 6 & 7 & 20 & 3978 & 25207 & Yes\\
%  \bottomrule[1pt]
%\end{tabular}
%%\begin{tablenotes}
%%        \item[*] The computing time is denoted by \textit{overtime} (o.t.) if the program does not terminate within 28,800 seconds (8 hours).
%%      \end{tablenotes}
%\end{threeparttable}
%}
%\end{table}

We use a parameterized plant (chosen from \cite{c3}) depicted in Fig. \ref{FMS} to test the efficacy and efficiency of our method in this section.
%Chosen from \cite{c3}, a marked net is depicted in Fig. \ref{FMS} as a parameterized net $\langle N, M_0\rangle$.
Let $M_0 = [\lambda\ 0\ 0\ 0\ 0\ 0\ 0\ \mu\ \mu\ \lambda\ 0\ 0\ 0\ 0\ 0\ 0\ \mu\ \lambda\ \mu\ \mu\ \mu\ \mu]^{\rm T}$.
%
%$M_0(p_1)=M_0(p_{10})=M_0(p_{18})=\lambda$ and $M_0(p_{8})=M_0(p_{9})=M_0(p_{17})=M_0(p_{19})=M_0(p_{20})=M_0(p_{21})=M_0(p_{22})=\mu$.
%Consider $T_E=\{t_3, t_6, t_{11}, t_{13}\}$ (marked as shadow bars) and $T_I=\{t_1, t_2, t_{4}, t_{5}, t_7, t_8, t_{9}, t_{10}, t_{12}, t_{14}, t_{15}, t_{16}\}$.
Consider $T_E=\{t_3, t_6, t_{11}, t_{13}\}$ (marked as shadow bars).
%Also, we set $\mathcal{M_F}=\mathcal{L}_{(\textbf{w},k)}=\{M\in\mathbb{N}^m| \textbf{w}^{\rm T}\cdot M \leq k\}$, where $\textbf{w}=[0\ 1\ 0\ 0\ 0\ 0\ 0\ 0\ 1\ 0\ 0\ 0\ 0\ 1\ 0\ 0\ 0\ 0\ 0\ 0\ 0\ 0]^{\rm T}$ and $k=\mu$, to test nonblockingness of this plant for all cases.
Also, we set $\mathcal{M_F}=\mathcal{L}_{(\textbf{w},k)}=\{M\in\mathbb{N}^m| \textbf{w}^{\rm T}\cdot M \leq k\}$, where $\textbf{w}=[0\ 0\ 0\ 0\ 0\ 0\ 0\ 0\ 0\ 0\ 0\ 1\ 1\ 1\ 1\ 1\ 0\ 0\ 0\ 0\ 0\ 0]^{\rm T}$ and $k=3$ (for run 4) or $k=4$ (for runs 8$\--$10) or $k=5$ (for runs 1$\--$3) or $k=7$ (for runs 5$\--$6) or $k = 15$ (for run 7), to test nonblockingness of this plant for all cases.

We run several simulations on a laptop with Intel i7-5500U 2.40 GHz processor and 8 GB RAM. Table \ref{tablenew1} shows, for different values of the parameters $\lambda$ and $\mu$, the sizes of the reachability graph $|R(N, M_0)|$, of the expanded BRG $|\mathcal{M_{B_E}}|$ \cite{Gu} and of minimax-BRG $|\mathcal{M_{B_M}}|$ as well as the time required to compute them. We also show the ratios of $|\mathcal{M_{B_M}}|$ to $|\mathcal{M_{B_E}}|$ and $|\mathcal{M_{B_M}}|$ to $|R(N, M_0)|$.
It can be verified that $|\mathcal{M_{B_M}}|\ll|\mathcal{M_{B_E}}|$ and $|\mathcal{M_{B_M}}|\ll|R(N, M_0)|$ in all cases.
%********************Analysis***********************
%{\blue Note that although the size of minimax-BRG grows more than linearly with the increasing of the system scale, it still does not grow exponentially as the size of the reachability graph with the increase of $\lambda$ and $\mu$, which depends on the net structure, initial resource distribution and choice of basis partition $\pi=(T_E, T_I)$.}
%***************************************************
Note that the size of minimax-BRG depends on the net structure, initial resource distribution and choice of basis partition $\pi=(T_E, T_I)$.
%{\blue Note that the size of minimax-BRG does not grow exponentially with the increase of $\lambda$ and $\mu$, which depends on the net structure, initial resource distribution and choice of basis partition $\pi=(T_E, T_I)$.}
%In Table \ref{tablenew2}, we show the performance of computing the set of non-final dead markings $\mathcal{D}_{\rm nf}$ based on Algorithm \ref{Algonew} and the total time required to determine nonblockingness (i.e., the sum of time to compute $\mathcal{M_{B_M}}$ and time to compute $\mathcal{D}_{\rm nf}$) in all cases.
%The cardinality of $\mathcal{D}_{\rm nf}$ and the time required to compute it for all cases are shown in columns 4$\--5$.
%Since $\mathcal{D}_{\rm nf}\neq\emptyset$ for all cases, we conclude in column 7 that the system is blocking in all cases.
%
%
%
%{\red Also in Table \ref{tablenew1}, we show the performance of determining if there exist non-final dead markings based on Algorithm \ref{Algonew} in columns 10$\--$11, and verifying unobstructiveness for all cases if necessary (with the usage of an ILPP solver namely \textit{LPSOVE} when computing $\mathcal{M_{\rm i_{co}}}$) in columns 12$\--$13.}
Also in Table \ref{tablenew1}, we show the simulation results of determining if there exist non-final dead markings based on Algorithm \ref{Algonew} (columns 10$\--$11), and verifying unobstructiveness (the set of i-coreachable markings $\mathcal{M_{\rm i_{co}}}$ of a minimax-BRG can be obtained by using two free MATLAB integer linear programming problems solver toolboxes namely YALMIP \cite{lofberg2004yalmip} and lpsolve \cite{berkelaar2004lpsolve}) for all cases if necessary (columns 12$\--$13).
%{\blue Also in Table \ref{tablenew1}, we show the performance of determining if there exist non-final dead markings based on Algorithm \ref{Algonew} in columns 10$\--$11, and verifying unobstructiveness for all cases if necessary in columns 12$\--$13. Note that the set of i-coreachable markings $\mathcal{M_{\rm i_{co}}}$ of a minimax-BRG can be obtained by leveraging two free MATLAB ILPP solver toolboxes namely YALMIP \cite{lofberg2004yalmip} and LPSOLVE.}
Moreover, the nonblockingness of the system for all cases are listed in column 14.
The test cases show that minimax-BRG-based technique achieves practical efficiency when coping with the NB-V problem in this considered case.
Additional case studies are also considered in \cite{GitHub}, which consists of three Petri net benchmarks taken from the literature.

\section{Discussions}
We propose the minimax-BRG to ensure that the essential features of a system, from which a blocking condition may originate, are captured in the abstracted model. As a non-trivial task, it is necessary to formally characterize and validate the proposed approach with a series of theoretical results.
When tackling the NB-V problem, the minimax-BRG-based approach is general and can be directly applied to arbitrary bounded plants (the only restriction is that the $T_I$-induced sub-net is acyclic). This is a major practical advantage with respect to other abstraction approaches that are based on particular structures or symmetries, and require significant analysis of the model in a preliminary stage before they can be applied.

Further, our numerical results (i.e., Section \ref{sec2.3} and \cite{GitHub}) show that the minimax-BRG can often be more compact in size than that of the reachability graph in the considered cases.
Accordingly, as a potential advantage, when it comes to a related problem of NB-V, i.e., \textit{nonblocking enforcement}, which consists of designing a \textit{supervisor} (an online control agent) to ensure that the controlled plant does not reach a blocking marking, a supervisor designed based on the minimax-BRG can also be more compact than that of a reachability-graph-based one.

%Note that the time required to detect whether those essential markings are non-final is negligible and, therefore, not listed.}
%As the scale of the system expands, the semi-structural method based on minimax-BRG is more efficient than the methods based on RG and expanded BRG.
%
%%On the other hand, the complexity of solving ILPP (\ref{CS1}) for each minimax-basis marking in the minimax-BRG of a net can be quite high if the net contains plenty of transitions with multiple input places; however, that complexity is linear in the number of nodes of the minimax-BRG.
\section{Conclusions and Future Work}\label{Con}
%************************************Original************************************************************************
%In this paper, we study the problem of nonblockingness verification of a plant.
%%We first define the unobstructiveness of a BRG, which can be verified by solving a set of ILPPs without constructing the reachability graph.
%A semi-structural method using minimax-BRG is developed, which can be used to determine the nonblockingness of a deadlock-free Petri net by checking its unobstructiveness.
%This approach is generalized to nets that are not deadlock-free by computing the set of non-final dead markings based on the set of minimax basis markings. Hence, to verify nonblockingness, one can first determine the existence of non-final deadlocks and later check the unobstructiveness of its minimax-BRG.
%The main advantages of our methods are that they do not require an exhaustive enumeration of the state space and have wide applicability.
%
%As for future work, we will investigate necessary and sufficient conditions for verifying nonblockingness in unbounded nets.
%Second, if a plant net is blocking, we plan to study the nonblockingness enforcement problem and develop a supervisor to guarantee the closed-loop system to be nonblocking.
%********************************************************************************************************************
In this paper, we studied the problem of nonblockingness verification of a plant.
%We first define the unobstructiveness of a BRG, which can be verified by solving a set of ILPPs without constructing the reachability graph.
A semi-structural method using minimax-BRG is developed, which can be used to determine the nonblockingness of a system modelled by bounded Petri nets by first determining the existence of non-final deadlocks and later checking the unobstructiveness of the corresponding minimax-BRG.
%{\red The main advantages of our methods are that they achieve practical efficiency and have wide applicability.}
The proposed approach does not require the construction of the reachability graph and has wide applicability.
%
%achieves practical efficiency in some considered cases.
As for future work, we will investigate necessary and sufficient conditions for verifying nonblockingness in unbounded nets.
Second, if a system is blocking, we plan to study the nonblockingness enforcement problem and develop a supervisor to guarantee the closed-loop system to be nonblocking.

\bibliographystyle{plain}        % Include this if you use bibtex
\bibliography{automatica2019}

\begin{thebibliography}{10}

\bibitem{berkelaar2004lpsolve}
M.~Berkelaar, K.~Eikland, and P.~Notebaert.
\newblock lpsolve: Open source (mixed-integer) linear programming system.
\newblock {\em Eindhoven U. of Technology}, 63, 2004.

\bibitem{Basis}
M.~P. Cabasino, A.~Giua, M.~Pocci, and C.~Seatzu.
\newblock Discrete event diagnosis using labeled {Petri} nets. an application
  to manufacturing systems.
\newblock {\em Control Engineering Practice}, 19(9):989--1001, 2011.

\bibitem{c8}
M.~P. Cabasino, A.~Giua, and C.~Seatzu.
\newblock Fault detection for discrete event systems using {Petri} nets with
  unobservable transitions.
\newblock {\em Automatica}, 46(9):1531--1539, 2010.

\bibitem{SL}
C.~G. Cassandras and S.~Lafortune.
\newblock {\em Introduction to discrete event systems}.
\newblock Springer, 2009.

\bibitem{ghaffari2003design}
A.~Ghaffari, N.~Rezg, and X.~L. Xie.
\newblock Design of a live and maximally permissive {Petri} net controller
  using the theory of regions.
\newblock {\em IEEE Transactions on Robotics and Automation}, 19(1):137--141,
  2003.

\bibitem{c1}
A.~Giua.
\newblock Supervisory control of {Petri} nets with language specifications.
\newblock In {\em Control of discrete-event systems}, pages 235--255. Springer,
  2013.

\bibitem{c11}
A.~Giua and F.~DiCesare.
\newblock Blocking and controllability of {Petri} nets in supervisory control.
\newblock {\em IEEE Transactions on Automatic Control}, 39(4):818--823, 1994.

\bibitem{c25}
A.~Giua, F.~DiCesare, and M.~Silva.
\newblock Generalized mutual exclusion contraints on nets with uncontrollable
  transitions.
\newblock In {\em Proceedings of the IEEE International Conference on Systems,
  Man, and Cybernetics}, pages 974--979. IEEE, 1992.

\bibitem{gohari2000complexity}
P.~Gohari and W.~M. Wonham.
\newblock On the complexity of supervisory control design in the {RW}
  framework.
\newblock {\em IEEE Transactions on Systems, Man, and Cybernetics, Part B:
  Cybernetics}, 30(5):643--652, 2000.

\bibitem{Gu}
C.~Gu, Z.~Y. Ma, Z.~W. Li, and A.~Giua.
\newblock Verification of nonblockingness in bounded {Petri} nets with a
  semi-structural approach.
\newblock In {\em Proceedings of the 58th IEEE Conference on Decision and
  Control}, pages 6718--6723. IEEE, 2019.

\bibitem{GitHub}
C.~Gu, Z.~Y. Ma, Z.~W. Li, and A.~Giua.
\newblock Simulations of the minimax basis reachability graph.
\newblock \url{https://github.com/ChaoGu92/Minimax-BRG}, accessed September,
  2020.

\bibitem{hu2015maximally}
H.~S. Hu, Y.~Liu, and M.~C. Zhou.
\newblock Maximally permissive distributed control of large scale automated
  manufacturing systems modeled with {Petri} nets.
\newblock {\em IEEE Transactions on Control Systems Technology},
  23(5):2026--2034, 2015.

\bibitem{leduc2005hierarchical}
R.~Leduc, B.~Brandin, M.~Lawford, and W.~M. Wonham.
\newblock Hierarchical interface-based supervisory control-part i: Serial case.
\newblock {\em IEEE Transactions on Automatic Control}, 50(9):1322--1335, 2005.

\bibitem{leduc2000hierarchical}
R.~Leduc, B.~Brandin, and W.~M. Wonham.
\newblock Hierarchical interface-based non-blocking verification.
\newblock In {\em Canadian Conference on Electrical and Computer Engineering},
  volume~1, pages 1--6. IEEE, 2000.

\bibitem{c20}
Z.~W. Li, M.~C. Zhou, and N.~Q. Wu.
\newblock A survey and comparison of {Petri} net-based deadlock prevention
  policies for flexible manufacturing systems.
\newblock {\em IEEE Transactions on Systems, Man, and Cybernetics, Part C:
  Applications and Reviews}, 38(2):173--188, 2008.

\bibitem{clin}
F.~Lin and W.~M. Wonham.
\newblock Verification of nonblocking in decentralized supervision.
\newblock {\em Control-Theory and Advanced Technology}, 7(1):223--232, 1991.

\bibitem{lofberg2004yalmip}
J.~Lofberg.
\newblock {YALMIP}: A toolbox for modeling and optimization in {MATLAB}.
\newblock In {\em 2004 IEEE international conference on robotics and automation
  (IEEE Cat. No. 04CH37508)}, pages 284--289. IEEE, 2004.

\bibitem{luo2009supervisor}
J.~L. Luo, W.~M. Wu, H.~Y. Su, and J.~Chu.
\newblock Supervisor synthesis for enforcing a class of generalized mutual
  exclusion constraints on {Petri} nets.
\newblock {\em IEEE Transactions on Systems, Man, and Cybernetics-Part A:
  Systems and Humans}, 39(6):1237--1246, 2009.

\bibitem{c24}
C.~Ma and W.~M. Wonham.
\newblock Nonblocking supervisory control of state tree structures.
\newblock {\em IEEE Transactions on Automatic Control}, 51(5):782--793, 2006.

\bibitem{c3}
Z.~Y. Ma, Y.~Tong, Z.~W. Li, and A.~Giua.
\newblock Basis marking representation of {Petri} net reachability spaces and
  its application to the reachability problem.
\newblock {\em IEEE Transactions on Automatic Control}, 62(3):1078--1093, 2016.

\bibitem{mohajerani2016framework}
S.~Mohajerani, R.~Malik, and M.~Fabian.
\newblock A framework for compositional nonblocking verification of extended
  finite-state machines.
\newblock {\em Discrete Event Dynamic Systems}, 26(1):33--84, 2016.

\bibitem{c27}
J.~Moody and P.~Antsaklis.
\newblock {\em Supervisory control of discrete event systems using {Petri}
  nets}, volume~8.
\newblock Springer, 2012.

\bibitem{Murata}
T.~Murata.
\newblock {Petri} nets: Properties, analysis and applications.
\newblock {\em Proceedings of the IEEE}, 77(4):541--580, 1989.

\bibitem{RW}
P.~J. Ramadge and W.~M. Wonham.
\newblock Supervisory control of a class of discrete event processes.
\newblock {\em SIAM Journal on Control and Optimization}, 25(1):206--230, 1987.

\bibitem{su2010aggregative}
R.~Su, J.~H.~Van Schuppen, and J.~E. Rooda.
\newblock Aggregative synthesis of distributed supervisors based on automaton
  abstraction.
\newblock {\em IEEE Transactions on Automatic Control}, 55(7):1627--1640, 2010.

\bibitem{c10}
Y.~Tong, Z.~W. Li, C.~Seatzu, and A.~Giua.
\newblock Verification of state-based opacity using {Petri} nets.
\newblock {\em IEEE Transactions on Automatic Control}, 62(6):2823--2837, 2016.

\bibitem{uzam2002optimal}
M.~Uzam.
\newblock An optimal deadlock prevention policy for flexible manufacturing
  systems using {Petri} net models with resources and the theory of regions.
\newblock {\em The International Journal of Advanced Manufacturing Technology},
  19(3):192--208, 2002.

\bibitem{wang2013design}
S.~G. Wang, C.~Y. Wang, and M.~C. Zhou.
\newblock Design of optimal monitor-based supervisors for a class of {Petri}
  nets with uncontrollable transitions.
\newblock {\em IEEE Transactions on Systems, Man, and Cybernetics: Systems},
  43(5):1248--1255, 2013.

\bibitem{RW3}
W.~M. Wonham and K.~Cai.
\newblock {\em Supervisory control of discrete-event systems}.
\newblock Springer, 2019.

\bibitem{zhao2013iterative}
M.~Zhao and Y.~F. Hou.
\newblock An iterative method for synthesizing non-blocking supervisors for a
  class of generalized {Petri} nets using mathematical programming.
\newblock {\em Discrete Event Dynamic Systems}, 23(1):3--26, 2013.

\end{thebibliography}

\end{document}